\documentclass[floatfix,twocolumn,pra,superscriptaddress,notitlepage,floatfix]{revtex4-2}
\usepackage[T1]{fontenc}
\usepackage[utf8]{inputenc}
\usepackage[english]{babel}

\usepackage[dvipsnames]{xcolor}
\usepackage{graphicx}
\usepackage{float}
\usepackage{tikz}
\usepackage{svg}
\usepackage{epstopdf}

\usepackage{amsmath}
\usepackage{amsthm}
\usepackage{amsfonts}
\usepackage{amssymb}
\usepackage{amsbsy}
\usepackage{mathtools}
\usepackage{bm}
\usepackage{bbm}
\usepackage{physics}
\usepackage{mathrsfs}
\usepackage{dsfont}

\usepackage{dcolumn}
\usepackage{makecell}
\usepackage{multirow}
\usepackage{longtable}

\usepackage[ruled,linesnumbered]{algorithm2e}
\usepackage{algpseudocode}

\usepackage{enumerate}
\usepackage{framed}
\usepackage{comment}
\usepackage{qcircuit}
\usepackage{faktor}
\usepackage{changepage}

\usepackage[colorlinks = true]{hyperref}
\hypersetup{colorlinks=true, linkcolor=blue, citecolor=blue, urlcolor=black}
\usepackage[capitalise]{cleveref}

\makeatletter
\newtheorem*{rep@theorem}{\rep@title}
\newcommand{\newreptheorem}[2]{%
\newenvironment{rep#1}[1]{%
 \def\rep@title{#2 \ref{##1}}%
 \begin{rep@theorem}}%
 {\end{rep@theorem}}}
\makeatother

\newtheorem{theorem}{Theorem}
\newtheorem{proposition}{Proposition}

\newtheorem{lemma}{Lemma}

\newtheorem*{assumption*}{Assumption}
\newtheorem{corollary}{Corollary}
\newtheorem{definition}{Definition}

\newcommand{\p}{\textsf{pauli}}

\newcommand{\ii}{\textup{i}}

\newcommand{\mc}[1]{\mathcal{#1}}

\begin{document}
\title{Exponentially reduced circuit depths in Lindbladian simulation}
\begin{abstract}
\end{abstract}
\date{\today}
\author{Wenjun Yu}
\affiliation{QICI Quantum Information and Computation Initiative, School of Computing and Data Science, The University of Hong Kong, Pokfulam Road, Hong Kong}

\author{Xiaogang Li}

\affiliation{Center on Frontiers of Computing Studies,  Peking University, Beijing 100871, China}
\affiliation{School of Computer Science, Peking University, Beijing 100871, China}

\author{Qi Zhao}
\email[]{zhaoqi@cs.hku.hk}
\affiliation{QICI Quantum Information and Computation Initiative, School of Computing and Data Science, The University of Hong Kong, Pokfulam Road, Hong Kong}
\author{Xiao Yuan}
\email{xiaoyuan@pku.edu.cn}

\affiliation{Center on Frontiers of Computing Studies,  Peking University, Beijing 100871, China}
\affiliation{School of Computer Science, Peking University, Beijing 100871, China}
\begin{abstract}

Quantum computers can efficiently simulate Lindbladian dynamics, enabling powerful applications in open system simulation, thermal and ground-state preparation, autonomous quantum error correction, dissipative engineering, and more. 
Despite the abundance of well-established algorithms for closed-system dynamics, simulating open quantum systems on digital quantum computers remains challenging due to the intrinsic requirement for non-unitary operations.
Existing methods face a critical trade-off: either relying on resource-intensive multi-qubit operations with experimentally challenging approaches or employing deep quantum circuits to suppress simulation errors using experimentally friendly methods.
In this work, we challenge this perceived trade-off by proposing an efficient Lindbladian simulation framework that minimizes circuit depths while remaining experimentally accessible. 
Based on the incoherent linear combination of superoperators, our method achieves exponential reductions in circuit depth using at most two ancilla qubits and the straightforward Trotter decomposition of the process. 
Furthermore, our approach extends to simulate time-dependent Lindbladian dynamics, achieving logarithmic dependence on the inverse accuracy for the first time.
Rigorous numerical simulations demonstrate clear advantages of our method over existing techniques. This work provides a practical and scalable solution for simulating open quantum systems on quantum devices.
\end{abstract}
\maketitle

\paragraph*{Introduction.---}
Quantum simulation of open quantum systems is one of the most promising applications of quantum computing.
To effectively describe these dynamic processes, researchers frequently employ the Lindblad master equation, which models large environments as memoryless baths under the Markovian approximation~\cite{lindblad1976generators,gorini1976completely}.
This framework provides sufficient generality to capture a wide range of open-system phenomena across various fields such as quantum optics~\cite{gardiner2004quantum,manzano2016atomic}, quantum biology~\cite{dorner2012towards,manzano2013quantum,plenio2008dephasing,mohseni2008environment}, and quantum chemistry~\cite{May_Volkhard2023-06-06,nitzan2006chemical}.
Beyond its role in scientific inquiry, Lindbladian simulation has proven transformative for quantum computing tasks, including thermal and ground-state preparation~\cite{cattaneo2023quantum,rost2021demonstrating,ding2024single,cubitt2023dissipative}, dissipative quantum state engineering~\cite{verstraete2009quantum}, autonomous quantum error correction~\cite{Reiter2017,Xu2023,PhysRevLett.128.020403}, and the solution of differential equations~\cite{shang2024design}.
The potential to simulate these complex dynamics positions Lindbladian simulations at the forefront of efforts to leverage quantum computing for advancing both foundational science and practical quantum technologies.

Despite its broad applicability, implementing Lindblad simulation poses significant technical challenges due to the fundamental difficulty of performing non-unitary evolution within the inherently unitary framework of quantum circuits.
Current approaches are generally divided into two categories, each with distinct limitations.
The first category comprises advanced-resource methods, which utilize linear combinations of unitaries to achieve polylogarithmic circuit depths~\cite{cleve2016efficient,li2022simulating,chen2023quantum,peng2024quantum}.
While these methods are theoretically efficient, they rely on sophisticated multi-qubit gates that exceed the capabilities of current hardware.
The second category includes methods that prioritize hardware accessibility, leveraging tools such as dilated Hamiltonians~\cite{childs2016efficient,cattaneo2021collision,pocrnic2023quantum,ding2024simulating} or Kraus expansions~\cite{schlimgen2021quantum,schlimgen2022quantum}. 
These approaches rely on more experimentally feasible quantum operations but are suboptimal in theory, necessitating polynomial circuit depth to suppress simulation errors.
This trade-off between implementation accessibility and circuit efficiency plausibly reveals the intrinsic difficulty of simulating system-environment coupling, a challenge that has long been considered fundamental to open quantum system dynamics.

\begin{figure}[t]
    \centering
    \includegraphics[width=\columnwidth]{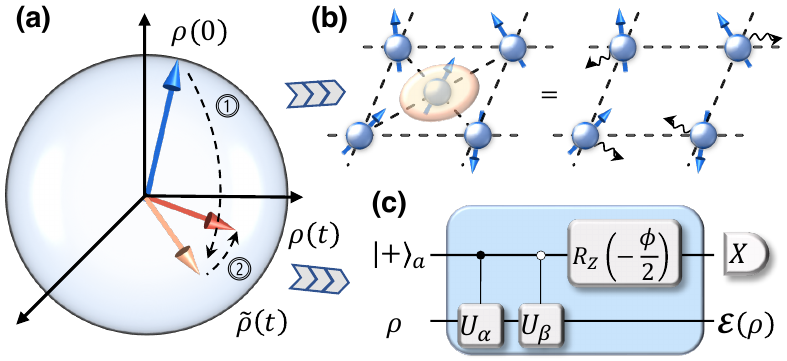}
    \caption{Schematic representation of our two-stage simulation framework. 
    (a) The hierarchical approach combines a coarse-grained simulation (Stage 1) with precision-enhancing compensation (Stage 2).
    The first stage constitutes the dominant part of the simulation, making the overall algorithm experimental-friendly.
     (b) Implementation of non-unitary dynamics in Stage 1, utilizing system-ancilla coupling for dissipative evolution in a dilated Hilbert space. 
     (c) Core circuit architecture for Stage 2, implementing unitary conjugates to achieve high-precision compensation.}
    \label{fig:concept}
\end{figure}

Here, we demonstrate that this trade-off can be overcome by proposing an efficient simulation protocol that achieves both polylogarithmic depths and minimal resource requirements, using no more than two ancilla qubits.
Our approach, depicted in Fig.~\ref{fig:concept}, employs a two-stage strategy: an accessible coarse-grained simulation followed by systematic compensation.
The coarse-grained stage accomplishes the primary part of our simulation,  leaving only small residual errors to be addressed.
Therefore, the compensation only needs to correct high-order discrepancies, ensuring both computational efficiency and high precision.
The key innovation lies in our linear combination of superoperators framework (LCS) to achieve compensation.
The LCS framework provides a comprehensive representation through a unitary-conjugate basis for all Hermitian-preserving linear maps, including operations for compensating the simulation errors.
Implementation is achieved through the circuit shown in Fig.~\ref{fig:concept}(c), which realizes the unitary-conjugate LCS formulas with minimal overhead.

We further extend our protocol to fulfill time-dependent Lindblad simulations.
Our method first discretizes the generic Lindbladians into piecewise constant segments.
This new evolution can be approximated by a time-independent simulation according to the preceding protocol.
Our LCS framework facilitates the further elimination of the remaining approximation error.
Therefore, we construct a time-dependent Lindbladian simulation protocol with exponentially improved accuracy for the first time.\\

\paragraph*{Settings.---}
The dynamics of open quantum systems, which interact with their environment, is described by the Lindblad master equation:
\begin{gather}
    \frac{d\rho}{dt}=\underbrace{\vphantom{\left(D_l \rho D_l^{\dagger}-\frac{1}{2}\{D^\dag_lD_l,\rho\}\right)}-\ii[H,\rho]}_{\mc H[\rho]}+\sum_{l=1}^m\underbrace{\left(D_l\rho D_l^\dag+\frac{1}{2}\{D^\dag_lD_l,\rho\}\right)}_{\mc D_l[\rho]},
\end{gather}
where $H$ represents the system's internal dynamics, and jump operators $\{D_l\}$ characterize how the system interacts with its environment.
The evolution can be expressed through a superoperator $\mc L\coloneqq\mc H+\sum_{l=1}^m\mc D_l$, which may exhibit continuous time dependence.

The challenge of Lindbladian simulation is to compute how an initial state $\rho(0)$ evolves under this dynamics to reach $\rho(T)=\mc T\exp(\int_0^T \mc L dt) \rho(0)$.
For most practical applications, it suffices to extract measurement outcomes of the final state rather than explicitly constructing it---an approach known as \emph{effective simulation}.
Our work presents an efficient way to effectively simulate the Lindbladian evolution.
For practical implementations, we focus on systems where both $H$ and $\{D_l\}$ are $q$-sparse, meaning they can be expressed using at most $q$ non-zero Pauli operators.
To quantify the evolution's strength, we define the Lindbladian norm as $\|\mc L\|_\p\coloneqq2(\|H\|_1+\sum_{l=1}^m\|D_l\|_1^2)$, where $\|\cdot\|_p$ measures the $p$th order vector norm of the operator's Pauli-decomposition coefficients.
\\

\paragraph*{LCS Formula.---}
Before presenting the algorithm, we first introduce the linear combination of superoperators (LCS) formula, which provides a powerful framework for analyzing and implementing a broad class of superoperators beyond physical operations.
Specifically, a LCS formula of a superoperator $\mc V$ is a convex combination of $\mc{V}_i$ with positive coefficients $c_i$ as
\begin{equation} \label{eq:LCS}
\mc V = \sum_{i=0}^{\Gamma-1} c_i \mc{V}_i = \mu \sum_{i=0}^{\Gamma-1} \Pr(i) \mc{V}_i,
\end{equation}
where $\mu=\sum_{i=0}^{\Gamma-1} c_i$, the normalization factor, is defined as the 1-norm and $\Pr(i)=c_i/\mu$.
The convex combination structure of the LCS formula in Eq.~\eqref{eq:LCS} naturally lends itself to implementation via Monte Carlo sampling.
This approach involves applying superoperator $\mc V_i$ to the quantum state with probability $\Pr(i)$, followed by a renormalization $\mu$ of the output state. The renormalization $\mu$ introduces additional sample overhead by amplifying statistical fluctuations to $\epsilon_s=\order{\mu/\sqrt{N}}$ in contrast
to $\order{1/\sqrt{N}}$ with $N$ rounds. 
For small constant $\mu$, however, the average of the quantum states still rapidly converges to the result of applying the desired superoperator. 
For tasks requiring the sequential application of multiple superoperators, we can concatenate their respective LCS formulae. This composition yields a new LCS formula that targets the overall superoperator.

While more powerful basis choices ${\mc V_i}$ can enhance performance, we prioritize experimental feasibility by utilizing minimal resources. 
To this end, we select \emph{unitary-conjugate} operations as our basis, defined as:
\begin{gather}\label{eq:UC}
    \mc V_{\phi_{\alpha\beta},\alpha,\beta}[\rho]\coloneqq\frac{1}{2}(\mathrm{e}^{\ii\phi_{\alpha\beta}} U_\alpha\rho U_\beta^\dag+\mathrm{e}^{-\ii\phi_{\alpha\beta}} U_\beta\rho U_\alpha^\dag),
\end{gather}
where $U_\alpha,U_\beta$ are unitary operators, and $\phi_{\alpha\beta}=-\phi_{\beta\alpha}$.
These operations can be efficiently implemented using the circuit structure depicted in Fig.~\ref{fig:concept} (c), requiring only a single ancilla qubit.
Moreover, we find this basis can express any Hermitian-preserving (HP) maps (Pauli-conjugate as a valid basis).
The HP assumption further implies $\Pr(\alpha,\beta)=\Pr(\beta,\alpha)$ as in~\cite{de1967linear,poluikis1981completely}.
When considering the concatenation of a series of LCS formulas, we show in the appendix that ancilla qubits can be reused without intermediate measurements.
This optimization reduces the quantum resource requirements while maintaining the desired operational sequence, enhancing the efficiency of our approach for complex quantum operations.\\

\begin{figure}[t]
    \centering
    \includegraphics[width=\columnwidth]{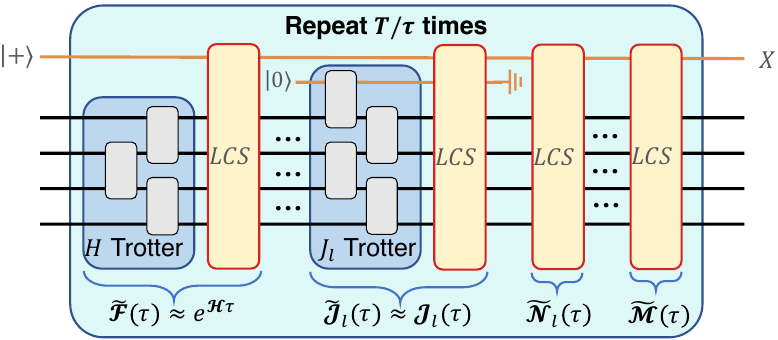}
    \caption{Circuit implementation diagram for a single simulation step. Blue blocks represent the Trotter-based simulation in the coarse-grained stage, while yellow blocks indicate error compensation operations. 
    The implementation requires at most two ancilla qubits, one of which can be reused for LCS formula execution as detailed in Appendix~\ref{sec:reuse}.}
    \label{fig:circuit}
\end{figure}

\paragraph*{Algorithm Overview.---}
We discretize the evolution time  $T$  into small intervals  $\tau \ll 1$ to control simulation errors.
For each time step, our approach to time-independent Lindblad adopts a hierarchical protocol: the coarse-grained simulation stage followed by the systematic error compensation stage, as shown in Fig.~\ref{fig:concept}.

The coarse-grained stage builds upon the conventional Lie-Trotter decomposition
\begin{gather}\label{eq:Lind_de}
    \mathrm{e}^{\mc L\tau}=\overleftarrow{\prod_{l=1}^m}\left(\mathrm{e}^{\mc D_l\tau}\right)\circ \mathrm{e}^{\mc H\tau}+\order{\tau^2},
\end{gather}
where the left arrow denotes descending order.
However, the inherent non-unitarity of the dissipative terms $\{\mathrm{e}^{\mc D_l\tau}\}$ still poses challenges to the realization.
To overcome this, we introduce a single-qubit-assisted joint Hamiltonian evolution $\mathrm{e}^{-\ii J_l\sqrt{\tau}}$, where
\begin{equation}
J_l\coloneqq\left(\begin{array}{cccc}
0 & D_l^{\dagger} \\
D_l & 0 
\end{array}\right)=\sigma_-\otimes D_l+\sigma_+\otimes D_l^\dag.
\end{equation}
This operation after tracing out the ancilla, denoted as $\mc J_{l}(\tau)$, approximates $\mathrm{e}^{\mc D_l\tau}$ to the first order:
\begin{equation}\label{eq:J_def}
 \Tr_a(\mathrm{e}^{-\ii J_l\sqrt{\tau}} [\ket{0}\bra{0}_a\otimes \rho] \mathrm{e}^{\ii J_l\sqrt{\tau}} ) = \rho + \tau\cdot\mc D_l[\rho] +\mc O(\tau^2).
\end{equation}
Both $J_l$ and $H$ evolutions are then approximated using the first-order Trotter-Suzuki formula, as shown in the blue blocks of Fig.~\ref{fig:circuit}.

This Trotter-based approach, while experimentally feasible, involves three distinct sources of error including the approximation of $\mc L$ evolution in Eq.~\eqref{eq:Lind_de}, the approximation of $\mc D_l$ evolution in Eq.~\eqref{eq:J_def}, and the Trotter errors in the evolution of $J_l$ and $H$.
Next, we show how to systematically compensate for these errors using the LCS of unitary-conjugate operations with minimal ancilla qubit overhead.
We summarize this process in yellow blocks in Fig.~\ref{fig:circuit}.\\

\paragraph*{Compensation Correction.---}
While first-order Trotter formulas provide a feasible approach for our coarse-grained simulation, a systematic compensation stage is essential to achieve high accuracy.
Here, we illustrate this novel compensation stage using LCS formulas.

To address higher-order errors in Eq.~\eqref{eq:Lind_de}, we desire a compensation $\mc M(\tau)$ that satisfies 
\begin{equation}\label{eq:Lind_decomp}
\begin{aligned}
		\mathrm{e}^{\mc L\tau}=\mc M(\tau) \circ \overleftarrow{\prod_{l=1}^m}\left(\mathrm{e}^{\mc D_l\tau}\right)\circ \mathrm{e}^{\mc H\tau}.
\end{aligned}
\end{equation}
Using exponential function's properties, we can express $\mc M(\tau)$ as $\mathrm{e}^{\mc L\tau} \circ\mathrm{e}^{-\mc H\tau}\circ\overrightarrow{\prod_{l=1}^m} \mathrm{e}^{-\mc D_l\tau}$.
When we expand each exponential term in its Taylor series, we obtain a power series $\mc M(\tau)=\sum_{k=0}^\infty \mc M_k\tau^k$ with
\begin{gather}\label{eq:M_k}
    \mc M_k\coloneqq\sum_{s+\sum_{l=0}^ms_l=k}\frac{(\mc H+\sum_{l=1}^m\mc D_l)^s(-\mc H)^{s_0}\overrightarrow{\prod_{l=1}^m}(-\mc D_l)^{s_l}}{s!\prod_{l=0}^m(s_l!)}.
\end{gather}

Although this compensation operation contains inverse dissipative terms that cannot be directly implemented physically, the LCS formula provides a practical solution.
By decomposing the Hamiltonian $\mc H$ and dissipative terms ${\mc D_l}$ in the Pauli-conjugate basis, we obtain an LCS formula for $\mc M(\tau)$. 
We then truncate $\mc M(\tau)$ to order $K$, yielding a finite-term approximation $\tilde{\mc M}(\tau)=\sum_{k=0}^K\mc M_k\tau^k$, as detailed in Appendix~\ref{sec:M}.
The truncated formula $\tilde{\mc M}(\tau)$ consists of Pauli conjugates with a quadratic 1-norm $\mu(\tilde{\mc M}(\tau))=\order{1+\| \mc L\|_\p^2\tau^2}$ (since $\mc M_1=0$) and an exponentially small bias
\begin{gather}
    \epsilon_M\coloneqq\|\tilde{\mc M}(\tau)-\mc M(\tau)\|_\diamond\leq\left(\frac{4e\|\mc L\|_\p\tau}{K+1}\right)^{K+1}.
\end{gather}

To implement this compensation, we employ the preceding Monte-Carlo sampling approach.
We first select an order $k_0\leq K$ with probability $\mu(\mc M_{k_0})\tau^{k_0}/\mu(\tilde{\mc M}(\tau))$.
We then sample a Paul-conjugate operator according to its weight in $\mc M_{k_0}$, which is calculated through Eq.~\eqref{eq:M_k} and the decomposition of $\mc H$ and $\{\mc D_l\}$.
As shown in Fig.~\ref{fig:concept}(c), the sampled operations can be implemented using two controlled Pauli gates, requiring one ancilla qubit and $\order{K}$ elementary gates for Pauli locality $\order{K}$.
Overall, we effectively suppress the Lie-Trotter error in Eq.~\eqref{eq:Lind_de} by applying $\tilde{\mc M}(\tau)$. 

The method extends similarly to address higher-order discrepancies in both dissipative and Hamiltonian evolutions.
For each case, we derive an explicit compensation operation, truncate to the first $K$th-order terms, and apply the same sampling process to realize the LCS formula.
We delay the construct ideas in Appendix~\ref{sec:N} and \ref{sec:H} for interested readers.
Each compensation similarly requires $\order{K}$ elementary gates and one additional ancilla qubit and brings an exponentially small bias with $K$.
Note that compensation for Hamiltonian evolution of $J_l$ in Eq.~\eqref{eq:J_def} requires two ancilla qubits since $J_l$ is already a $n+1$-qubit Hamiltonian operator.\\

\paragraph*{Simulation Construction and Overheads.---}
We now construct a comprehensive high-accuracy simulation of Lindblad evolution by integrating previously described two-stage processes.
For each step, the simulation takes the form:
\begin{gather}
    \mathrm{e}^{\mc L\tau}\approx\tilde{\mc M}(\tau) \circ \overleftarrow{\prod_{l=1}^m}\left(\tilde{\mc N}_{l}(\tau)\tilde{\mc J}_{l}(\tau)\right)\circ \tilde{\mc F}(\tau),
\end{gather}
where $\tilde{\mc N}_l(\tau)$ compensates for dissipative term errors, while $\tilde{\mc J}_{l}(\tau)$ and $\tilde{\mc F}(\tau)$ represent compensations for Trotter evolutions of $J_l$ and $H$, as illustrated in Fig.~\ref{fig:circuit}.

To achieve high accuracy while maintaining efficiency, we make several key design choices. First, we set the truncation order:
\begin{gather}\label{eq:K}
    K=\Theta\left(\frac{\log(\varepsilon\tau/T)}{\log(\|\mc L\|_\p \tau)}\right)
\end{gather}
to ensure each step's bias remains below $\varepsilon\tau/T$.
Second, we choose $\tau=\order{1/(\|\mc L\|_\p^2 T)}$ to keep the overall 1-norm $\mu=\exp[\order{\|\mc L\|_\p^2\tau T}]$ at $\Theta(1)$ when concatenating $T/\tau$ steps.
Under these parameters, each step requires $\order{m(q+K)}$ elementary gates due to the first-order Trotter's $\order{q}$ gate cost.
These choices lead to our main result (proved in Appendix~\ref{sec:independent}):

\begin{theorem}\label{thm:time-independent}
    For a $q$-sparse Lindblad operator $\mc L$, simulation time $T$, target error $\varepsilon$, the gate and sample complexities for our Lindbladian simulation algorithm are
    \begin{equation}
        \begin{aligned}
            N_g &= \order{(\|\mc L\|_\p T)^2m\left(q+\frac{\log(1/\varepsilon)}{\log(\|\mc L\|_\p T)}\right)},\\
N&=\Theta\left(1/\varepsilon^2\right)
        \end{aligned}
    \end{equation}
\end{theorem}
\noindent Therefore, without the sophisticated or expensive constructions in Refs.~\cite{cleve2016efficient,li2022simulating,chen2023quantum,peng2024quantum}, we also achieved logarithmic dependence on $1/\varepsilon$ for open system simulation under the experimentally more friendly Trotter framework. \\

\paragraph*{Time-dependent Lindblad Simulation.---}
We can extend our approach with compensation to time-dependent Lindblad simulation.
For each time step $[t_0, t_0+\tau)$, we construct a piecewise constant Lindbladian $\mc L_M(t)$ by partitioning the step into $M$ equal slices.
Within a slice $j$, $\mc L_M(t)$ takes the value $\mc L_j$, evaluated at the slice's midpoint.
When $\mc L(t)$ has a bounded derivative, $\|\dot{\mc L}\|_\diamond\coloneqq\max_t\|\dot{\mc L}(t)\|_\diamond$, this new evolution approximates:
\begin{align}\label{eq:discrete_error}
\|\mc T\mathrm{e}^{\int_{t_0}^{t_0+\tau}\mc L_M(t)dt}-\mc T\mathrm{e}^{\int_{t_0}^{t_0+\tau}\mc L(t)dt}\|_\diamond=\order{\frac{\tau^2\|\dot{\mc L}\|_\diamond}{M}}.
\end{align}

We implement the evolution under $\mc L_M(t)$ through a time-independent simulation $\mathrm{e}^{\sum_{j=1}^M\mc L_j\tau/M}$ followed by a compensation operator:
\begin{align}\label{eq:W_def}
    \mc W(\tau)\coloneqq\mc T\mathrm{e}^{\int_{t_0}^{t_0+\tau}\mc L_M(t)dt}\circ\mathrm{e}^{-\sum_{j=1}^M\mc L_j\tau/M}.
\end{align}
The time-independent part can be realized using our previous analysis, so we focus on $\mc W(\tau)$ here.
Using Dyson and Taylor expansions along with the Pauli-conjugate basis 
$\mc W(\tau)$ can be decomposed as a convergent power series when $\tau\ll1/\|\mc L\|_\p$, where $\|\mc L\|_\p\coloneqq\max_t\|\mc L(t)\|_\p$.
We show the derivation in Appendix~\ref{sec:depend}.
Truncating the series at order $K$ yields a finite Pauli-conjugate LCS formula $\tilde{\mc W}(\tau)$ with 1-norm $\exp[\order{(\|\mc L\|_\p\tau)^3}]$ and bias:
\begin{align}\label{eq:W_error}
    \epsilon_W\coloneqq\|\tilde{\mc W}(\tau)-\mc W(\tau)\|_\diamond\leq\left(\frac{2\mathrm{e}\|\mc L\|_\p\tau}{K+1}\right)^{K+1}.
\end{align}

The single-step simulation thus becomes:
\begin{gather}\label{eq:dependent-step}
    \mc T\mathrm{e}^{\int_{t_0}^{t_0+\tau}\mc L(t)dt}\approx\tilde{\mc W}(\tau)\circ\mathrm{e}^{\sum_{j=1}^M\mc L_j\tau/M},
\end{gather}
of which the gate count increases only by the cost of $\tilde{\mc W}(\tau)$ and still remains $\order{m(q+K)}$ given the sparsity assumption.
To achieve the full simulation, we iterate the short-time evolution $T/\tau$ times, where single-step error stays within $\order{\varepsilon\tau/T}$.
Three key parameters control the simulation accuracy and efficiency:
First, we set $\tau=\order{1/(\|\mc L\|_\p^2 T)}$ which similarly results in $\mu=\order{1}$ even accounting for  the additional factor from $\tilde{\mc W}(\tau)$.
Second, we choose $M=\Theta\left(\|\dot{\mc L}\|_\diamond/(\varepsilon\|\mc L\|^2)\right)$ to bound the discretization error in Eq.~\eqref{eq:discrete_error}.
Finally, by adopting the same order $K$ as in Eq.~\eqref{eq:K}, we suppress both bias in Eq.~\eqref{eq:W_error} and the time-independent error in Eq.~\eqref{eq:dependent-step}.
These choices collectively guarantee the following analysis:
\begin{theorem}
    For a $q$-sparse time-dependent Lindblad operator $\mc L$, simulation time $T$, target error $\varepsilon$, the gate and sample complexities for our Lindbladian simulation algorithm are
    \begin{equation}
        \begin{aligned}
            N_g &= \order{(\|\mc L\|_\p T)^2m\left(q+\frac{\log(1/\varepsilon)}{\log(\|\mc L\|_\p T)}\right)},\\
N&=\Theta\left(1/\varepsilon^2\right)
        \end{aligned}
    \end{equation}
\end{theorem}
In summary, our compensation-based algorithm applies to time-dependent Lindblad simulation and achieves polylogarithmic depths for the first time. \\

\begin{figure*}[t]
    \centering
    \includegraphics[width=0.95\linewidth]{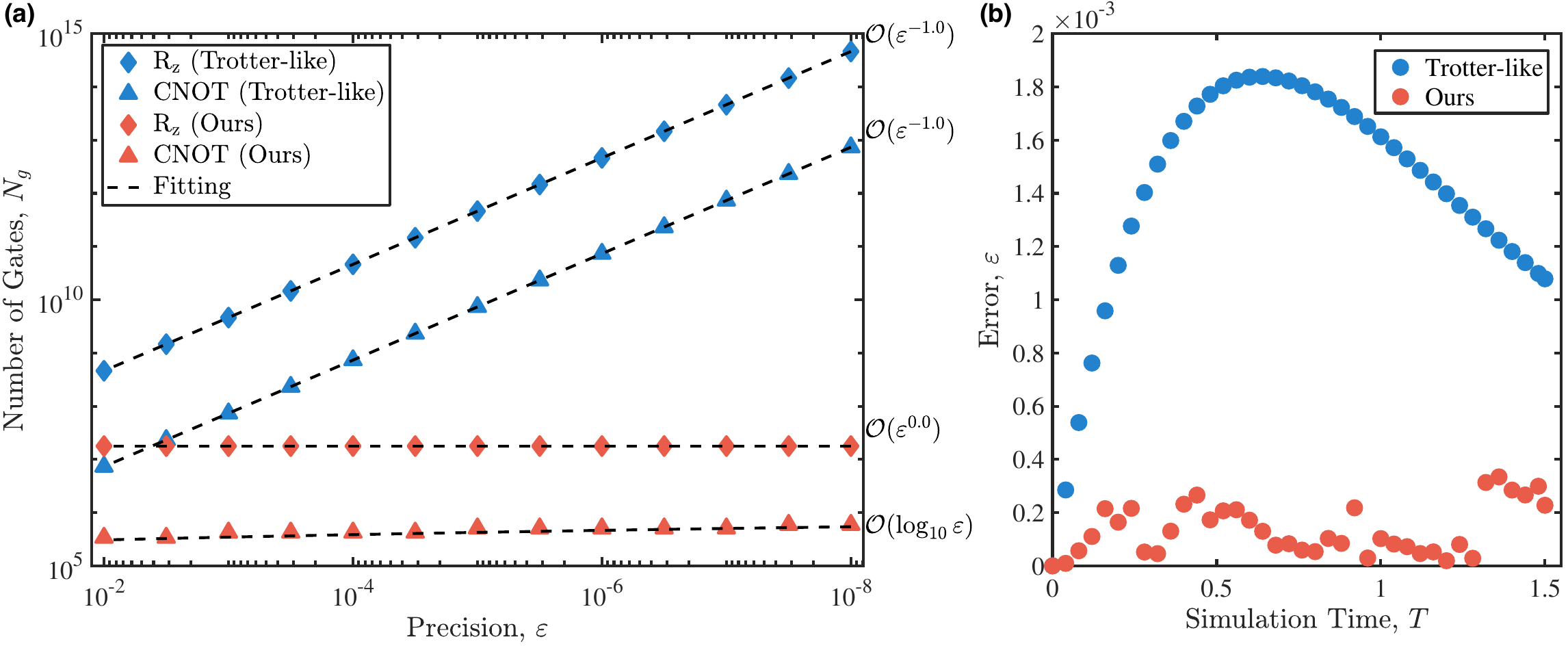}
    \caption{Numerical comparison of our simulation method versus conventional Trotter decomposition for a TFI model with localized dissipation.
    (a) Gate count comparison for a 20-qubit system, showing scaling of $R_z (\theta) = \mathrm{e}^{\ii\theta Z}$ and CNOT gates versus diamond-norm precision.
    Gate counts are derived from theoretical bounds for both methods, with time steps or compensation orders adjusted to achieve specified precision.
    (b) Real-time simulation errors for a five-qubit system ($J=-0.1$, $h=0.2$, $\gamma=1.5$).
    Both simulations track evolution from initial state $\ket{\psi_0}=\ket{10000}$ with time step $\tau=0.02$, measuring projection onto $\ket{\psi_0}\bra{\psi_0}$.
Our method employs $5\times10^6$ Monte-Carlo samples.
    }
    \label{fig:numerical}
\end{figure*}
\paragraph*{Numerical Results.---}
To substantiate our theoretical framework, we conduct comprehensive numerical simulations on a paradigmatic quantum system: a transverse-field Ising (TFI) model with periodic boundary conditions, described by the Hamiltonian:
$H= J\sum_{\langle i,j\rangle} {Z_iZ_j}+h\sum_{j}{X_j}$. 
To incorporate open system dynamics, we introduce dissipation via a jump operator $D=\sqrt{\gamma}\ket{0}\bra{1}$ acting on the first qubit.

Our numerical results, presented in Fig.~\ref{fig:numerical}, demonstrate the efficiency and accuracy of our approach. 
Subplot (a) presents a comparison of the resource requirements between our method and the conventional first-order Trotter algorithm--the later essentially implementing only our coarse-grained stage. 
We compile circuits into elementary gate sets comprising CNOT gates, single-qubit rotations, and single-qubit Clifford gates.
Focusing on the resource-intensive CNOT and rotation gates, our analysis demonstrates that our method achieves exponentially better gate efficiency as precision requirements become more stringent.
We also evaluate the simulation accuracy through real-time evolution in Fig.~\ref{fig:numerical}(b). 
Empirical implementations demonstrate that our approach consistently maintains nearly an order of magnitude improvement in accuracy throughout the simulation.
These strongly support the practicality of our theoretical framework for high-precision open system simulation.

\paragraph*{Conclusion and Outlook.---}
We have developed a Lindblad simulation method that achieves polylogarithmic circuit depths while requiring no more than two ancilla qubits. This breakthrough disproves the presumed trade-off between quantum resource usage and efficiency in solving the Lindblad simulation problem. Furthermore, our method can be extended to simulate time-dependent Lindblad dynamics, achieving for the first time a polylogarithmic dependence on the desired simulation error.
Despite these advancements, there is still room for further optimization. For instance, higher-order methods could be employed in the coarse-grained stages to reduce the complexity associated with extended simulation times. Additionally, an open question remains: can the current multiplicative polylogarithmic complexity be improved to an additive logarithmic complexity? Addressing these challenges could pave the way for even more efficient quantum simulations. \\

\begin{acknowledgements}
We thank Pei Zeng for the helpful discussions.
W.Y. and Q.Z. acknowledge funding from National Natural Science Foundation of China (NSFC) via Project No. 12347104 and No. 12305030, Guangdong Natural Science Fund via Project 2023A1515012185, Hong Kong Research Grant Council (RGC) via No. 27300823, N\_HKU718/23, and R6010-23, Guangdong Provincial Quantum Science Strategic Initiative GDZX2200001, HKU Seed Fund for Basic Research for New Staff via Project 2201100596. 
X.L. and X.Y. are supported by 
the National Natural Science Foundation of China Grant (No.~12361161602 and No.~12175003), NSAF (Grant No.~U2330201), the Innovation Program for Quantum Science and Technology (Grant No.~2023ZD0300200).
The numerics is supported by the High-performance Computing Platform of Peking University.\\
  
  \emph{Note:} While completing this work, we became aware of a related study~\cite{kato2024exponentially}, which also leverages LCS-formula compensation to improve the precision of Lindblad simulations. However, our approach differs in several important ways. Their method does not exploit the decomposition of dissipative terms, a critical strategy for minimizing the number of required ancilla qubits. Moreover, their simulation framework focuses on time-independent Lindbladian superoperators, whereas our approach applies more generally to time-dependent cases.
\end{acknowledgements}

\bibliographystyle{apsrev4-1}
\bibliography{ref}

\setcounter{theorem}{0}
\setcounter{lemma}{0}
\setcounter{proposition}{0}
\setcounter{definition}{0}
\setcounter{corollary}{0}
\setcounter{algocf}{0}
\clearpage
\onecolumngrid
\appendix
\section{Preliminaries}
\subsection{Representation of Superoperators}
In the scope of quantum information processing, we usually deal with states of quantum systems and operations therein.
There are two major families of representations for these ingredients.
For an isolated quantum system, which we refer to as a closed system, its state can be fully described by a set of basis and a corresponding vector in the Hilbert space $\mathscr{H}_n$.
The operations summarized by Schr\"odinger equation and the intrinsic linearity are denoted by multiplying by unitary operators (matrices) in $\mathrm{L}(\mathscr{H}_n)$.
In the remainder of this paper, we will use capital letters to denote these operators.
Nevertheless, it is ubiquitous to find the target quantum system correlating with some environmental systems, which we name an open system.
In this case, states of the system are generically described by density operators (matrices) in $\mathrm{L}(\mathscr{H}_n)$, which are semi-definite Hermitian square matrices with unit traces.
The operations on these open-system states are, therefore, named as \emph{superoperators}.
In the remainder of this paper, we will use calligraphic letters to denote these operators.

To represent the superoperator formally, we make use of the linearity of quantum operations and devise several methods.
For simplicity, we introduce the Pauli process-matrix representation for an arbitrary superoperator
\begin{definition}\label{def:HP_Pauli}
    For an arbitrary linear map $\mc E: \mathrm{L}(\mathscr{H}_n)\rightarrow\mathrm{L}(\mathscr{H}_n)$, we define the Pauli process-matrix representation of $\mc E$ as
    \begin{gather}\label{eq:Pauli_component}
       \mc E[\rho]=\sum_{\alpha,\beta\in{\sf P}^n}\chi_{\alpha\beta}P_\alpha\rho P_\beta^\dag,
   \end{gather}
   where ${\sf P}^n$ consists of $4^n$ Pauli indices, and $P_\alpha, P_\beta$ are $n$-qubit Pauli operators. In this decomposition, each process matrix $\chi$ element denotes a complex coefficient. 
   We further denote $\abs{\chi}\coloneqq\sum_{\alpha,\beta\in{\sf P}^n}\abs{\chi_{\alpha\beta}}$.
\end{definition}

\subsection{Open-System Simulations with Lindblad Master Equation}\label{sec:open}
In this work, we aim to simulate the dynamics of open-system evolution.
To specify the settings of the simulation problem, we start from the most general time-dependent Lindblad master equation.
When referring to the time-dependent Lindblad equation, we consider a bounded continuous Superoperator-valued function $\mc L(t)$ defined on $t\in[0,T]$.
The corresponding Markovian open-system dynamics on a state $\rho$ at time $t_0\in[0,T]$ is depicted by
\begin{align}\label{eq:Lind}
\frac{\mathrm{d}\rho}{\mathrm{d}t}\Big|_{t_0}=\underbrace{\vphantom{\sum_{l=1}^m\left(D_l \rho D_l^{\dagger}-\frac{1}{2} D_l^{\dagger} D_l \rho-\frac{1}{2} \rho D_l^{\dagger} D_l\right)}-i[H(t_0), \rho]}_{\mc H(t_0)}+\underbrace{\sum_{l=1}^m\left(D_l(t_0) \rho D_l(t_0)^{\dagger}-\frac{1}{2} D_l(t_0)^{\dagger} D_l(t_0) \rho-\frac{1}{2} \rho D_l(t_0)^{\dagger} D_l(t_0)\right)}_{\mc D(t_0)}=\mc L(t_0)[\rho],
\end{align}
where $H(t_0)$ and $D_l(t_0)$'s are continuous Hamiltonian and jump operators, respectively.

The Lindblad master equation leads to a clear formalism of the open-system simulation task: Given a Hamiltonian $H(t)$, jump operators $\{D_l(t)\}$, an input state $\rho(0)$, and a time $T$, we want to (approximately) output the final state:
\begin{gather}
    \rho(T)\coloneqq\mc T\mathrm{e}^{\int_{0}^T\mc L(t)dt}[\rho(0)]=\mc T\mathrm{e}^{\int_{0}^T(\mc \mc H(t)+\mc D(t))dt}[\rho(0)],
\end{gather}
where $\mc T$ represents time-ordered operator.
Nevertheless, extracting some properties of the final state is often more relevant than getting the state explicitly.
This generates a compromised version of the simulation task which we denote by the \emph{effective simulation}:  Given a Hamiltonian $H(t)$, jump operators $\{D_l(t)\}$, an input state $\rho(0)$, and a time $T$, we want to estimate some property $O$ of the final state $\Tr(\mc T\mathrm{e}^{\int_{0}^T\mc L(t)dt}[\rho(0)]O)$.
In this work, we focus on the effective simulation of a continuous Lindbladian superoperator with sparse local Hamiltonian and jump operators.

To elucidate these constraints, we first consider the sparsity assumption.
To quantify the sparsity of the Hamiltonian and jump operator, we introduce the family of \emph{Pauli-induced operator norms}.
For an arbitrary operator with Pauli decomposition: $A=\sum_{\alpha\in{\sf P}^n} c_\alpha P_\alpha$, we define the Pauli-induced operator norms based on the vector norm of the Pauli decomposition coefficients: $\|A\|_{p}\coloneqq\|\vec{c}\|_p$.
Based on this, the sparsity constraint requires that $\|H\|_0\leq q$ and $\|D_l\|_0\leq q$ for all $l\in[m]$, implying no more than $q$ non-zero Pauli components.
Furthermore, the locality assumption imposes that those relevant Pauli components are local.

To further consider the superoperator induced by these Hamiltonian and jump operators, we adopt the Pauli process-matrix representation in Def.~\ref{def:HP_Pauli}.
\begin{gather}\label{eq:Lind_chi}
        \mc H(t)=\sum_{\alpha,\beta}\chi_{\alpha\beta}^{(0)}(t)P_\alpha[\cdot]P_\beta^\dag,\hspace{2em}
        \mc D_l=\sum_{\alpha,\beta}\chi_{\alpha\beta}^{(l)}(t)P_\alpha[\cdot]P_\beta^\dag,\ \forall\, l\in[m].
\end{gather}
By the triangle inequality, we have the following conditions on these superoperators:
\begin{gather}\label{eq:1norm_lind}
    \|\mc H(t)\|_\diamond\leq\abs{\chi^{(0)}(t)}\leq2\|H(t)\|_1,\hspace{2em} \|\mc D_{l}(t)\|_\diamond\leq\abs{\chi^{(l)}(t)}\leq2\|D_l(t)\|_1^2,\ \forall\, l\in[m].
\end{gather}
Our problem setting also imposes that the Lindbladian is continuous.
Therefore, we can use the following norm to quantify the Lindbladian's intensity.
\begin{gather}
    \|\mc L\|_\p\coloneqq\max_t2\left(\|H(t)\|_1+\sum_{l=1}^m\|D_l(t)\|_1^2\right).
\end{gather}

\subsection{Linear Combination of Superoperators Formula}\label{sec:LCS}
As a power tool to analyze and realize a broader class of superoperators beyond CPTP maps, we introduce the following linear combination of super-operators (LCS) formula.
\begin{definition}\label{def:lcs}
    A $(\mu,\varepsilon)$-linear-combination-of-super-operator (LCS) formula of a super-operator $\mc V$ is defined to be 
\begin{equation} \label{eq:LCS_app}
\tilde{\mc V} = \sum_{i=0}^{\Gamma-1} c_i \mc{V}_i = \mu \sum_{i=0}^{\Gamma-1} \Pr(i) \mc{V}_i,
\end{equation}
such that the diamond norm distance $\|\mc{V}-\tilde{\mc V}\|_\diamond\leq \varepsilon$, and $\max_i\|\mc{V}_i\|_\diamond=1$. 
We require every $c_i$ a positive real number and absorb the phase into $\mc{V}_i$.
Thus, $\mu\coloneqq\sum_{i=0}^{\Gamma-1}c_i$, and $\Pr(i)\coloneqq c_i/\mu$ is the probability of super-operators $\mc{V}_i$. 
We call $\mu$ the $1$-norm of this $(\mu,\varepsilon)$-LCS formula. 
\end{definition}

Given the definition of the formula in Eq.~\eqref{eq:LCS_app}, the target map can be effectively implemented with a $\varepsilon$ precision.
To achieve this, we employ a Monte-Carlo sampling over the decomposition $\{\mc V_i\}$, as summarized in Alg.~\ref{alg:random_sample}. 
After applying the sampled superoperator, we measure the state on a desired observable $O$ and re-normalize the outcome by the 1-norm $\mu$.
The results from the preceding post-processing are supposed to recover the map.

This Monte-Carlo routine to realize the LCS formula seems to require that all the basis superoperators $\{\mc V_i\}$ must be physical to be applied
Nevertheless, we will show that every sampled $\mc V_i$ can also be effectively implemented.
We validate this observation by introducing the \emph{unitary-conjugate} basis:
\begin{gather}
    \mc V_{\phi,\alpha,\beta}[\rho]\coloneqq\frac{1}{2}(\mathrm{e}^{\ii\phi} U_\alpha\rho U_\beta^\dag+\mathrm{e}^{-\ii\phi} U_\beta\rho U_\alpha^\dag).
\end{gather}
This family of maps is in general non-CPTP, which cannot be applied directly in a quantum system.
On the other hand, each map can be effectively implemented in the following circuit with only one ancilla qubit.
\begin{figure}[t]
    \centering
    \[ \Qcircuit @C=2em @R=1em {
  \lstick{\ket{+}} & \ctrl{1} & \ctrlo{1} &  \gate{R_Z(-\phi/2)} &\measureD{X}\\
  \lstick{\rho} & \gate{U_\alpha} & \gate{U_\beta} & \qw & \mc \mc V_{\phi,\alpha,\beta}[\rho]
}\]
    \caption{The circuit to effectively implement each unittary conjugate.}
    \label{fig:LCS}
\end{figure}
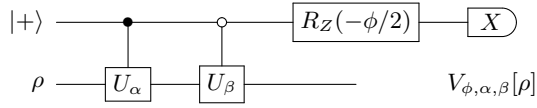
Since the subsequent Monte-Carlo sampling also realizes the formula in an effective manner, we can sample and implement this circuit for our unitary-conjugate formula.
Moreover, we also impose the constraint that $\Pr(\alpha,\beta)=\Pr(\beta,\alpha)$ and $\phi_{\alpha,\beta}=-\phi_{\beta,\alpha}$.

In the scope of this work and most quantum computing tasks, we exclusively focus on the Hermitian-Preserving linear maps.
To cater to the experimental implementation, we recruit a specific unitary-conjugate strategy intimately related to the Pauli process-matrix representation of the target map.
Based on~\cite{de1967linear,poluikis1981completely}, every Hermitian-preserving map $\mc E$ has a Hermitian Pauli process matrix $\chi$.
Therefore, we can decompose it via Def.~\ref{def:HP_Pauli} as
\begin{align}\label{eq:HP_decomp}
    \mc E=\sum_{\alpha,\beta}\chi_{\alpha\beta}P_\alpha[\cdot]P_\beta^\dag=\sum_{\alpha,\beta}\frac{\chi_{\alpha\beta}}{2}P_\alpha[\cdot]P_\beta^\dag+\frac{\chi^\star_{\alpha\beta}}{2}P_\beta[\cdot]P_\alpha^\dag=\sum_{\alpha,\beta}|\chi_{\alpha\beta}|\mc V_{\phi_{\alpha,\beta},\alpha,\beta}=\mu(\mc E)\sum_{\alpha,\beta}\Pr(\alpha,\beta)\mc V_{\phi_{\alpha,\beta},\alpha,\beta},
\end{align}
where $\mu(\mc E)=\abs{\chi}$, $\Pr(\alpha,\beta)=\abs{\chi_{\alpha\beta}}/\mu(\mc E)$, and $\phi_{\alpha,\beta}=\arctan{\frac{\Im \chi_{\alpha\beta}}{\Re\chi_{\alpha\beta}}}=-\phi_{\beta,\alpha}$.
Each basis operation is
\begin{gather}
    \mc V_{\phi,\alpha,\beta}(\rho)\coloneqq\frac{1}{2}(\mathrm{e}^{\ii\phi} P_\alpha\rho P_\beta^\dag+\mathrm{e}^{-\ii\phi} P_\beta\rho P_\alpha^\dag).
\end{gather}
With $\|\mc V_{\phi,\alpha,\beta}\|_\diamond=1$, Eq.~\eqref{eq:HP_decomp} forms a $(\mu(\mc E),0)$-LCS formula of $\mc E$ according to Def.~\ref{def:lcs}.
We call this the \emph{Pauli-conjugate} LCS formula.

When adopting the LCS formula to represent operations in the algorithm, we will inevitably encounter concatenations.
We show that the concatenation of LCS formulas will naturally result in a composed formula:
\begin{lemma}\label{lm:concate}
    Suppose we have a $(\mu_1,\epsilon_1)$-LCS formula, $\tilde{\mc U}$, for a superoperator $\mc U$ and a $(\mu_2,\epsilon_2)$-LCS formula, $\tilde{\mc V}$, for $\mc V$.
    The composition $\tilde{\mc V}\circ\tilde{\mc U}$ is a $(\mu_1\mu_2,\mu_1\epsilon_2+\mu_2\epsilon_1+\epsilon_1\epsilon_2)$-LCS formula for the superoperator $\mc V\circ\mc U$.
    Particularly, the bias can be reduced to $\epsilon_1+\epsilon_2+\epsilon_1\epsilon_2$ when both $\mc U$ and $\mc V$ are CPTP.
\end{lemma}
\begin{proof}
Since the concatenation of convex combinations is still a convex combination, the resultant from the composition of LCS formulas is a valid LCS formula.
The statement of the 1-norm is trivially true. 
We will derive the bias here.
    \begin{align}\label{eq:composition}
        \|\tilde{\mc V}\circ\tilde{\mc U}-\mc V\circ\mc U\|_\diamond
        \leq&\|\tilde{\mc V}\circ\tilde{\mc U}-\tilde{\mc V}\circ\mc U\|_\diamond+\| \tilde{\mc V}\circ\mc U-\mc V\circ\mc U\|_\diamond
        \leq\|\tilde{\mc U}-\mc U\|_\diamond\|\tilde{\mc V}\|_\diamond+\|\tilde{\mc V}-\mc V\|_\diamond\|\mc U\|_\diamond\notag\\
        \leq& \epsilon_1\mu_2+\epsilon_2(\mu_1+\epsilon_1).
        \end{align}
        For the special case where both $\mc U$ and $\mc V$ are CPTP, they have unit diamond norms.
        Thus, we have $\|\tilde{\mc V}\|_\diamond\leq1+\epsilon_2$.
        Substitute this into Eq.~\eqref{eq:composition}, and we can prove the claim.
\end{proof}

\begin{algorithm}[t]
    \caption{Random Sampling of LCS}\label{alg:random_sample}
    \KwIn{An LCS formula $\{\Pr(i),\mc V_i\}_i$; An input state $\rho_0$}
    \KwOut{An effectively evolved state}
    Sample $i$ according to the probability distribution $\{\Pr(i)\}$\;
    Apply $\mc V_i$ to $\rho_0$\;
    \Return $\mc V_{i}[\rho_0]$
\end{algorithm}

\section{Time-Independent Lindblad Simulation}
Our approach to Lindblad simulation comprises two distinct stages: a coarse-grained simulation using experimentally feasible methods, followed by systematic error compensation to eliminate higher-order inaccuracies. This two-stage strategy effectively balances implementation practicality with simulation precision.

In each time step with length $\tau$, the foundation of our coarse-grained stage rests on the Lie-Trotter decomposition:
\begin{gather}\label{eq:Lind_de_app}
    \mathrm{e}^{\mc L\tau}=\overleftarrow{\prod_{l=1}^m}\left(\mathrm{e}^{\mc D_l\tau}\right)\circ \mathrm{e}^{\mc H\tau}+\order{\tau^2},
\end{gather}
where the directional product notation denotes superoperator concatenation:
\begin{gather}
    \overrightarrow{\prod_{l=1}^m}\mc A_l=\mc A_1\circ\cdots\circ\mc A_m, \hspace{0.5em}\text{and}\hspace{0.5em}\overleftarrow{\prod_{l=1}^m}\mc A_l=\mc A_m\circ\cdots\circ\mc A_1.
\end{gather}

This decomposition reveals two essential components: coherent Hamiltonian simulation and dissipative simulation.
In the remainder of the section, we first detail how we can improve the Lie-Trotter approximation in Eq.~\eqref{eq:Lind_de_app} by our LCS-based compensation.
Moreover, we then elucidate how our approach implements these two major components with high accuracy respectively.
These subroutines can be employed collectively to facilitate the whole simulation.

\subsection{Compensation for Lie-Trotter Errors}\label{sec:M}
The Lie-Trotter error from Eq.~\eqref{eq:Lind_de_app} clearly contains a second-order leading term, which significantly constrains the accuracy of the overall simulation.
To overcome this inaccuracy, our method has utilized the following compensation $\mc M(\tau)$ such that
\begin{equation}\label{eq:Lind_decomp_ideal}
\begin{aligned}
		\mathrm{e}^{\mc L\tau}=\mc M(\tau) \circ \overleftarrow{\prod_{l=1}^m}\left(\mathrm{e}^{\mc D_l\tau}\right)\circ \mathrm{e}^{\mc H\tau}.
\end{aligned}
\end{equation}
Therefore, we can derive the explicit form of this compensation as
\begin{equation}\label{eq:M_def}
	\mc M(\tau) \coloneqq \mathrm{e}^{\mc L\tau} \circ\mathrm{e}^{-\mc H\tau}\circ\overrightarrow{\prod_{l=1}^m} \mathrm{e}^{-\mc D_l\tau}.
\end{equation}

It is noteworthy that this compensation contains several inverse dissipative evolution, which can not be physically implemented.
We instead adopt the LCS-formula method to effectively realize this compensation.
In the following, we will elaborate on how to construct the specific formula that approximates $\mc M(\tau)$ and uses the easily implementable Pauli conjugates.

\begin{proposition}\label{prop:trotter}
    The compensation $\mc M(\tau)$ in Eq.~\eqref{eq:M_def} can be constructed as a $(\mu,\epsilon)$-LCS formula on the Pauli-conjugate basis such that 
    \begin{align*}
        \mu\leq \mathrm{e}^{16(\|H\|_1+\sum_{l=1}^m\|D_l\|_1^2)^2\tau^2},\hspace{0.5em}\text{and}\hspace{0.5em}
        \epsilon\leq\left(\frac{4e(\|H\|_1+\sum_{l=1}^m\|D_l\|_1^2)\tau}{K_M+1}\right)^{K_M+1},
    \end{align*}
    where $K_M$ is the truncation order.
    The gate count for this formula is $\order{K_M}$ using one ancilla qubit.
\end{proposition}
\begin{proof}
    Based on Taylor's expansion, we decompose $\mc M(\tau)$ based on the superoperators $\mc H$ and $\{\mc D_l\}$:
\begin{align}\label{eq:M}
    \mc M(\tau)=\mathrm{e}^{\mc L\tau}\circ \mathrm{e}^{-\mc H\tau}\circ\overrightarrow{\prod_{l=1}^m}\mathrm{e}^{-\mc D_l\tau}=\sum_{s=0}^\infty\frac{(\mc H+\sum_{l=1}^m\mc D_l)^s\tau^s}{s!}\sum_{s_0=0}^\infty\frac{(-\mc H)^{s_0}\tau^{s_0}}{s_0!}\overrightarrow{\prod_{l=1}^m}\sum_{s_l=0}^\infty\frac{(-\mc D_l)^{s_l}\tau^{s_l}}{s_l!}=\sum_{k=0}^\infty \tau^k \mc M_k.
\end{align}
It is clear to see that $\mc M_0=\mc I$ and $\mc M_1=0$.
Each higher-order $\mc M_k$ based on is identified as a summation of Pauli conjugates:
\begin{align}\label{eq:M_decomp}
    \mc M_k\coloneqq&\sum_{s+\sum_{l=0}^ms_l=k}\frac{(\mc H+\sum_{l=1}^m\mc D_l)^s(-\mc H)^{s_0}\overrightarrow{\prod_{l=1}^m}(-\mc D_l)^{s_l}}{s!\prod_{l=0}^m(s_l!)}\notag\\
    =&\sum_{s+\sum_{l=0}^ms_l=k}\frac{1}{s!\prod_{l=0}^m(s_l!)}\sum_{\Vec{j}\in\{0,\cdots,m\}^s}\left(\overrightarrow{\prod_{i=1}^s}\mc E_{j_i}\right)(-\mc H)^{s_0}\overrightarrow{\prod_{l=1}^m}(-\mc D_l)^{s_l}\notag\\
    =&\sum_{s+\sum_{l=0}^ms_l=k}\frac{(-1)^{k-s}}{s!\prod_{l=0}^m(s_l!)}\sum_{\Vec{j}\in\{0,\cdots,m\}^s}\left(\overrightarrow{\prod_{i=1}^s}\sum_{\alpha,\beta}\chi_{\alpha\beta}^{(j_i)}P_\alpha(\cdot)P_\beta^\dag\right)\circ\overrightarrow{\prod_{l=0}^m}\left(\sum_{\alpha,\beta}\chi_{\alpha\beta}^{(l)}P_\alpha(\cdot)P_\beta^\dag\right)^{\circ s_l}\notag\\
    =&\sum_{s+\sum_{l=0}^ms_l=k}\frac{(-1)^{k-s}}{s!\prod_{l=0}^m(s_l!)}\sum_{\Vec{j}\in\{0,\cdots,m\}^s} \left(\overrightarrow{\prod_{i=1}^s}\,\abs{\chi^{(j_i)}}\sum_{\alpha,\beta}\frac{\chi^{(j_i)}_{\alpha\beta}}{\abs{\chi^{(j_i)}}}P_\alpha(\cdot)P_\beta^\dag\right)\circ\overrightarrow{\prod_{l=0}^m}\left(\abs{\chi^{(l)}}\sum_{\alpha,\beta}\frac{\chi^{(l)}_{\alpha\beta}}{\abs{\chi^{(l)}}}P_\alpha(\cdot)P_\beta^\dag\right)^{\circ s_l}\notag\\
    =&\mu(\mc M_k)\sum_{s+\sum_{l=0}^ms_l=k}\Pr(\Vec{s}\,|\,k)\sum_{\Vec{j}\in\{0,\cdots,m\}^s}\Pr(\Vec{j}\,|\,\Vec{s},k)\sum_{\Vec{\alpha},\Vec{\beta}}\Pr(\Vec{\alpha},\Vec{\beta}\,|\,\Vec{j},\Vec{s},k)\mc V_{(\Vec{s},\Vec{j}),\Vec{\alpha},\Vec{\beta}},
\end{align}
where we denote $\mc E_0=\mc H$ and $\mc E_l=\mc D_l$ for all $l\in[m]$.
Similarly, we recruit $\Vec{s}\coloneqq(s,s_0,\cdots,s_m)$, $\Vec{j}\coloneqq(j_1,\cdots,j_s)$, $\Vec{\alpha}\coloneqq(\alpha_{1},\cdots,\alpha_{k})$, $P_{\Vec{\alpha}}\coloneqq P_{\alpha_{1}}\cdots P_{\alpha_{k}}$, $\Vec{\beta}\coloneqq(\beta_{1},\cdots,\beta_{k})$, and $P_{\Vec{\beta}}\coloneqq P_{\beta_{1}}\cdots P_{\beta_{k}}$.
The map $\mc V_{(\Vec{s},\Vec{j}),\Vec{\alpha},\Vec{\beta}}$ is a Pauli conjugate for $P_{\Vec{\alpha}}$ and $P_{\Vec{\beta}}^\dag$ with the overall phase equal to the summation of phases from all $\chi$ matrices and the sign of $(-1)^{k-s}$.
The normalization factor is $\mu(\mc M_k)=\frac{(2\sum_{l=0}^m\abs{\chi^{(l)}})^k}{k!}$. 
The probability $\Pr(\Vec{s}\,|\,k)$ is a multinomial distribution such that
\begin{align}
    \Pr(\Vec{s}\,|\,k)=\Pr(s,s_0,\cdots,s_m\,|\,k)\coloneqq\frac{k!}{s!\prod_{l=0}^m(s_l)!}\frac{\prod_{l=0}^m\abs{\chi^{(l)}}^{s_l}}{2^k(\sum_{l=0}^m\abs{\chi^{(l)}})^{k-s}}.
\end{align}
The probability $\Pr(\Vec{j}\,|\,\Vec{s},k)$ can be viewed as the product of $s$ independent distributions:
\begin{gather}
    \Pr(\Vec{j}\,|\,\Vec{s},k)\coloneqq\prod_{i=1}^s\frac{\abs{\chi^{(j_i)}}}{\sum_{l=0}^m\abs{\chi^{(l)}}}.
\end{gather}
Moreover, the probability $\Pr(\Vec{\alpha},\Vec{\beta}\,|\,\Vec{j},\Vec{s},k)$ can also be decomposed as the product of $k$ independent distributions about the corresponding Pauli indices:
\begin{align}
    \Pr(\Vec{\alpha},\Vec{\beta}\,|\,\Vec{j},\Vec{s},k)\coloneqq\prod_{i=1}^s\frac{\chi^{(j_i)}_{\alpha\beta}}{\abs{\chi^{(j_i)}}}\prod_{l=0}^m\prod_{i=1}^{s_{l}}\frac{\chi^{(l)}_{\alpha_{s+\sum_{x=0}^{l-1}s_x+i}\beta_{s+\sum_{x=0}^{l-1}s_x+i}}}{\abs{\chi^{(l)}}}.
\end{align}

Regarding practical concerns, we always need to truncate the expansion in Eq.~\eqref{eq:M} to a fixed order $K_{M}$.
We denote the resulting summation by $\tilde{\mc M}(\tau)$:
\begin{gather}\label{eq:LCS_M}
    \tilde{\mc M}(\tau)\coloneqq\sum_{k=0}^{K_M}\tau^k\mc M_k=\mu(\tilde{\mc M}(\tau))\left(\Pr(0)\mc I+\sum_{k=2}^{K_M}\Pr(k)\frac{\mc M_k}{\mu(\mc M_k)}\right),
\end{gather}
where the 1-norm is
\begin{align}\label{eq:1-norm_M}
    \mu(\tilde{\mc M}(\tau))\coloneqq&1+\sum_{k=2}^{K_M}\tau^k\mu(\mc M_k)=1+\sum_{k=2}^{K_M}\frac{(2\sum_{l=0}^m\abs{\chi^{(l)}}\tau)^k}{k!}=\mathrm{e}^{2\sum_{l=0}^m\abs{\chi^{(l)}}\tau}-2\sum_{l=0}^m\abs{\chi^{(l)}}\tau\notag\\
    \leq&\mathrm{e}^{(2\sum_{l=0}^m\abs{\chi^{(l)}}\tau)^2}\leq\mathrm{e}^{16(\|H\|_1+\sum_{l=1}^m\|D_l\|_1^2)^2\tau^2},
\end{align}
where the last inequality comes from Eq.~\eqref{eq:1norm_lind}.
The probabilities satisfy $\Pr(0)=1/\mu(\tilde{\mc M}(\tau))$ and $\Pr(k)=\mu(\mc M_k)\tau^k/\mu(\tilde{\mc M}(\tau))$.

This truncation will introduce an additional error in the approximation toward $\mc M(\tau)$.
Fortunately, we can bound this error and show that it is benign.
\begin{align}
    \|\tilde{\mc M}(\tau)-\mc M(\tau)\|_\diamond=&\left\|\sum_{k=K_M+1}^\infty \tau^k\mc M_k\right\|_\diamond\leq\sum_{k=K_M+1}^\infty\tau^k\|\mc M_k\|_\diamond\notag\\
    \leq& \sum_{k=K_M+1}^\infty\tau^k\mu(\mc M_k)\sum_{s+\sum_{l=0}^ms_l=k}\Pr(\Vec{s}\,|\,k)\sum_{\Vec{j}\in\{0,\cdots,m\}^s}\Pr(\Vec{j}\,|\,\Vec{s},k)\sum_{\Vec{\alpha},\Vec{\beta}}\Pr(\Vec{\alpha},\Vec{\beta}\,|\,\Vec{j},\Vec{s},k)\|\mc V_{\phi^{\Vec{s},\Vec{j}}_{\Vec{\alpha},\Vec{\beta}},\Vec{\alpha},\Vec{\beta}}\|_\diamond\notag\\
    =&\sum_{k=K_M+1}^\infty\frac{(2\sum_{l=0}^m\abs{\chi^{(l)}}\tau)^k}{k!}\leq \left(\frac{4e(\|H\|_1+\sum_{l=1}^m\|D_l\|_1^2)\tau}{K_M+1}\right)^{K_M+1}.
\end{align}

To implement this $\tilde{\mc M}(\tau)$, we can use the circuit in Sec.~\ref{sec:LCS}.
Given that $\tilde{\mc M}(\tau)$ is an LCS formula with the Pauli-conjugate basis, the circuit requires one ancilla qubit and two control-Pauli gates.
Moreover, the Pauli terms in the LCS can achieve the locality of $\order{K_M}$, which implies $\order{K_M}$ elementary gates. 

\end{proof}

\subsection{High-Accuracy Dissipative Simulations}\label{sec:N}
The dissipative simulation is inherently challenging to quantum computing since it in general requires non-unitary operations.
In this section, we propose a two-stage method to realize this part:
After achieving a coarse-grained simulation via the dilated Hamiltonian idea, we continue an additional LCS formula to improve accuracy.

To limit the dimension of the ancillary system used, we narrow our focus to every single jump operator as in Eq.
\begin{gather}\label{eq:Ll}
   \mc D_l\coloneqq D_l [\cdot] D_l^{\dagger}-\frac{1}{2} D_l^{\dagger} D_l [\cdot]-\frac{1}{2} [\cdot] D_l^{\dagger} D_l,\ \forall\,l\in[m].
\end{gather} 
To achieve the coarse-grained simulation of $\mathrm{e}^{\mc D_l \tau}$, we recruit only one ancilla qubit and devise a joint Hamiltonian evolution $\mathrm{e}^{-\ii J_l\sqrt{\tau}}$ with
\begin{equation}
J_l\coloneqq\left(\begin{array}{cccc}
0 & D_l^{\dagger} \\
D_l & 0 
\end{array}\right)=\sigma_-\otimes D_l+\sigma_+\otimes D_l^\dag.
\end{equation}
For any non-negative integer $k$, the direct calculation shows: 
\begin{equation}
	\begin{aligned}
		J_l^{2k+1} &=  \ket{0}\bra{1}_a \otimes  \left(D_l D_l^\dag\right)^k D_l^\dag+ \ket{1}\bra{0}_a\otimes D_l \left(D_l^\dag D_l\right)^k,  \\
		J_l^{2k+2} &=  \ket{1}\bra{1}_a \otimes D_l\left(D_l^\dag D_l\right)^kD_l^\dag + \ket{0}\bra{0}_a \otimes \left(D_l^\dag D_l\right)^{k+1},
	\end{aligned}
\end{equation}
where the subscript $a$ indicates the ancillary system.
Then we can get the matrix exponential according to Taylor's expansion.
\begin{equation}
\begin{aligned}
		\mathrm{e}^{-\ii J_l\sqrt{\tau}}\ket{0}_a = \sum_{k = 0}^\infty \frac{(-i\sqrt{\tau})^{2k}}{(2k)!}\ket{0}_a\otimes \left(D_l^\dag D_l\right)^k  + \sum_{k = 0}^\infty \frac{(-i\sqrt{\tau})^{2k+1}}{(2k+1)!} \ket{1}_a\otimes  D_l \left(D_l^\dag D_l\right)^k.
\end{aligned}
\end{equation}
Therefore, the overall dynamical map has the effect, 
\begin{equation}\label{eq:J}
	\begin{aligned}
		\Tr_a[\mathrm{e}^{-\ii J_l\sqrt{\tau}} &[\ket{0}\bra{0}_a\otimes \rho] \mathrm{e}^{\ii J_l\sqrt{\tau}} ] \\
        =&\mc I +\sum_{k=1}^{\infty}\tau^k\left(\sum_{j=0}^k\frac{(-1)^k}{(2j)!(2k-2j)!}\mc D_{l,2}^j\circ\mc D_{l,3}^{k-j}+\sum_{j=0}^{k-1}\frac{(-1)^{k-1}}{(2j+1)!(2k-2j-1)!}\mc D_{l,1}\circ\mc D_{l,2}^{j}\circ\mc D_{l,3}^{k-j-1}\right)[\rho]\\
        =& \mc I+\sum_{k=1}^{\infty}\tau^k\sum_{j_1,\cdots,j_k=1}^3b^{(k)}_{j_1,\cdots,j_k}\mc D_{l,j_1}\circ\cdots\circ\mc D_{l,j_k}[\rho].
	\end{aligned}
\end{equation}
Here we used
\begin{align}\label{eq:L1}
    \mc D_{l,1}\coloneqq D_l[\cdot]D_l^\dag,\ \mc D_{l,2}\coloneqq D_l^\dag D_l[\cdot]\mc I,\ \mc D_{l,3}\coloneqq\mc I[\cdot]D_l^\dag D_l,\ \forall\, l\in[m],
\end{align}
and denotes coefficients in the series by $\{b_{j_1,\cdots,j_k}^{(k)}\}$.
We find the operation in Eq.~\eqref{eq:J} a first-order approximation given small $\tau$, which serves as a good coarse-grained step:
\begin{equation}\label{eq:J_def_app}
	\mc J_{l}(\tau)[\rho]\coloneqq\Tr_a[\mathrm{e}^{-\ii J_l\sqrt{\tau}} [\ket{0}\bra{0}_a\otimes \rho] \mathrm{e}^{\ii J_l\sqrt{\tau}} ] = \rho + \tau\left(D_l\rho D_l^\dag-\frac{1}{2}\{D_l^\dag D_l,\rho\}\right) +\mc O(\tau^2).
\end{equation}

To further improve the approximation from the coarse simulation $\mc J_{l}(\tau)$, we insert an additional superoperator $\mc N_{l}(\tau)$ to compensate for the discrepancies such that
\begin{align}\label{eq:N_def}
    \mc N_{l}(\tau)\mc J_{l}(\tau)=\mathrm{e}^{\mc D_l\tau}=\mc I+\sum_{k=1}^{\infty}\frac{\tau^k}{k!}\left(\mc D_{l,1}-\frac{1}{2}\mc D_{l,2}-\frac{1}{2}\mc D_{l,3}\right)^k
    =\mc I+\sum_{k=1}^{\infty}\tau^k\sum_{j_1,\cdots,j_k=1}^3c^{(k)}_{j_1,\cdots,j_k}\overrightarrow{\prod_{i=1}^k}\mc D_{l,j_i}.
\end{align}
Here $\{c^{(k)}_{j_1,\cdots,j_k}\}$ represent coefficients in the series of $\mathrm{e}^{\mc D_l\tau}$.
Since both $\mc J_{l}(\tau)$ and $\mathrm{e}^{\mc D_l\tau}$ can be decomposed on $\{\mc D_{l,j}\}_{j=1}^3$, we can represent this compensation as
\begin{align}\label{eq:N_exp}
    \mc N_{l}(\tau)\coloneqq\mc I+\sum_{k=1}^{\infty}\tau^k\sum_{j_1,\cdots,j_k=1}^3a^{(k)}_{j_1,\cdots,j_k}\overrightarrow{\prod_{i=1}^k}\mc D_{l,j_i}=\sum_{k=0}^\infty \tau^k\mc N_{l,k}.
\end{align}
Consequently, we can recursively determine the desired coefficients $\{a\}$ as.
\begin{gather}\label{eq:recur_def}
    a^{(k)}_{j_1,\cdots,j_k}=c^{(k)}_{j_1,\cdots,j_k}-\sum_{i=1}^{k-1} a^{(k-i)}_{j_1,\cdots,j_{k-i}}\cdot b^{(i)}_{j_{k-i+1},\cdots,j_k}-b^{(k)}_{j_1,\cdots,j_k},\ \forall k\in\mathbb{N}_+.
\end{gather}
It is easy to check that all of the three series $\{a\},\{b\}$ and $\{c\}$ are composed of real numbers.
Moreover, since the structure of the superoperators fully determines these coefficients, they are independent of the choice of $l$.
According to the recursion in Eq.~\eqref{eq:recur_def}, we can show a geometrically scaling upper bound on the 1-norm $a_k$ of the coefficients in the order of $k$.
\begin{lemma}\label{lm:geom}
    The 1-norm of the coefficients $\{a_{j_1,\cdots,j_k}^{(k)}\}$ in each order of $k$, $a_k\coloneqq\sum_{j_1,\cdots,j_{k_0}=1}^3\lvert a_{j_1,\cdots,j_{k_0}}^{(k_0)}\rvert$, is bounded by a geometric series with some constant ratio $\lambda>1$.
\end{lemma}
\begin{proof}
    According to a simple calculation, the norms of the first four orders are 1, 0, 1.5833, and 2.0472, respectively.
    We can make an inductive assumption that there exists an order-independent constant $\lambda\geq1$ such that
    \begin{gather}
        a_k\leq \lambda^{k-2}\cdot a_{2},\ \forall\,3\leq k< k_0,\,k\in\mathbb{N}_+.
    \end{gather}
    To complete this induction, we show that this bound also holds for $k=k_0$.
    \begin{align}
        a_{k_0}\coloneqq\sum_{j_1,\cdots,j_{k_0}=1}^3\lvert a_{j_1,\cdots,j_{k_0}}^{(k_0)}\rvert\leq&\sum_{j_1,\cdots,j_{k_0}=1}^3\lvert c_{j_1,\cdots,j_{k_0}}^{(k_0)}\rvert+\sum_{j_1,\cdots,j_{k_0}=1}^3\lvert b_{j_1,\cdots,j_{k_0}}^{(k_0)}\rvert+\sum_{i=1}^{k_0-1}\max(b^{(k_0-i)})\cdot\sum_{j_1,\cdots,j_i}\lvert a_{j_1,\cdots,j_i}^{(i)}\rvert\notag\\
        \leq& \frac{2^{k_0}}{k_0!}+\frac{2k_0}{(k_0!)^2}+\sum_{i=1}^{k_0-1}\frac{\lambda^{i-2}}{(k_0-i)!^2}a_2
        \leq 1+\sum_{x=1}^{k_0-1}\frac{\lambda^{k_0-2}}{x!^2\cdot\lambda^x}a_2\leq 2(\mathrm{e}^{\frac{1}{\lambda}}-1)\cdot\lambda^{k_0-2}a_2,
    \end{align}
    where the second inequality comes from the definition in Eq.~\eqref{eq:J} and the induction assumption.
    Therefore, it is easy to see that this induction holds for all $\lambda>2.5$. 
    Note that $\lambda$ is independent of the choice of $l$.
\end{proof}
The compensation superoperator $\mc N_{l}(\tau)$ is in general not a CPTP map, which renders it impossible to directly implement it using quantum operations.
Fortunately, we discovered that $\mc N_{l}(\tau)$ can be represented by a Pauli-conjugate LCS formula and thus be realized by sampling.
\begin{proposition}\label{prop:N_LCS}
    Suppose the time length is small enough that $\tau\leq1/(2\lambda\|D_l\|_1^2)$.
    The compensation $\mc N_{l}(\tau)$ in Eq.~\eqref{eq:N_exp} can be represented as a $(\mu,\epsilon)$-LCS formula on Pauli-conjugate basis such that 
    \begin{align*}
        \mu\leq\mathrm{e}^{2a_2\|D_l\|_1^4\tau^2},\hspace{0.5em}\text{and}\hspace{0.5em}
        \epsilon\leq2a_2\lambda^{K_l-1}\|D_l\|_1^{2K_l+2}\tau^{K_l+1},\notag
    \end{align*}
    where $a_2$ and $\lambda$ are defined as in Lemma~\ref{lm:geom}. 
    The gate count is $\order{K_l}$ using one ancilla qubit.
\end{proposition}
\begin{proof}

Since the compensation $\mc N_{l}(\tau)$ is defined on the basis of $\{\mc D_{l,i}\}$,
we begin by consider the Pauli process representations of $\{\mc D_{l,i}\}$
\begin{align}\label{eq:L1_chi}
    \mc D_{l,i}=\sum_{\alpha}\chi_{\alpha\beta}^{(l,i)}P_\alpha[\cdot]P_\beta^\dag,\ \forall\, l\in[m],\,i\in[3],
\end{align}
We can bound the $\chi$ matrix by the Pauli 1-norm as
\begin{gather}\label{eq:1normL1}
    \abs{\chi^{(l,i)}}\leq\|D_l\|_1^2,\ \forall\, l\in[m],\,i\in[3].
\end{gather}
Based on the observation
\begin{gather}
    \mc D_{l,1}[A]^\dag=\mc D_{l,1}[A^\dag],\ \mc D_{l,2}[A]^\dag=\mc D_{l,3}[A^\dag],\ 
    \mc D_{l,3}[A]^\dag=\mc D_{l,2}[A^\dag],\ \forall\, l\in[m],
\end{gather}
we know that $\mc D_{l,j}^\dag=\mc D_{l,\pi(j)}$ where $\pi=(1)(23)\in \mathrm{S}_3$ is a permutation and $\pi^{-1}=\pi$.
Moreover, we find that
\begin{align}
    a^{(k_0)}_{\pi(j_1),\cdots,\pi(j_{k_0})}=&c^{(k_0)}_{\pi(j_1),\cdots,\pi(j_{k_0})}-\sum_{i=1}^{k_0-1}a^{(k_0-i)}_{\pi(j_1),\cdots,\pi(j_{k_0-i})}\cdot b^{(i)}_{\pi(j_{k_0-i+1}),\cdots,\pi(j_{k_0})}-b^{(k_0)}_{\pi(j_1),\cdots,\pi(j_{k_0})}\notag\\
    =&c^{(k_0)}_{j_1,\cdots,j_{k_0}}-\sum_{i=1}^{k_0-1}a^{(k_0-i)}_{j_1,\cdots,j_{k_0-i}}\cdot b^{(i)}_{j_{k_0-i+1},\cdots,j_{k_0}}-b^{(k_0)}_{j_1,\cdots,j_{k_0}}=a^{(k_0)}_{j_1,\cdots,j_{k_0}}.
\end{align}

According to the coefficients $\{a\}$, we find that $\mc N_{l,0}=\mc I$ and $\mc N_{l,1}=0$.
Besides, we can further construct the Pauli-conjugate decomposition for higher-order $\mc N_{l,\tau}$.
\begin{align}\label{eq:N_HP}
   \mc N_{l,k}=&\sum_{j_1,\cdots,j_k=1}^3\frac{1}{2}a^{(k)}_{j_1,\cdots,j_k}\left(\overrightarrow{\prod_{i=1}^k}\mc D_{l,j_i}+\overrightarrow{\prod_{i=1}^k}\mc D_{l,\pi(j_i)}\right)\notag\\
    =&\sum_{j_1,\cdots,j_k=1}^3\frac{1}{2}a^{(k)}_{j_1,\cdots,j_k}\left(\overrightarrow{\prod_{i=1}^k}\sum_{\alpha,\beta}\chi^{(l,j_i)}_{\alpha\beta}P_{\alpha}[\cdot]P_{\beta}^\dag+\overrightarrow{\prod_{i=1}^k}\sum_{\alpha,\beta}\chi^{(l,j_i)\star}_{\alpha\beta}P_{\beta}[\cdot]P_{\alpha}^\dag\right)\notag\\
    =&\sum_{\Vec{j}\in\mathbb{F}_3^k}\abs{a^{(k)}_{\Vec{j}}}\left(\prod_{i=1}^k\abs{\chi^{(l,j_i)}}\right)\Bigg(\sum_{\Vec{\alpha},\Vec{\beta}}\Pr(\Vec{\alpha},\Vec{\beta}\,|\,l,\Vec{j})\frac{1}{2}\Big(\exp(\ii\phi^{l,\Vec{j}}_{\Vec{\alpha},\Vec{\beta}}) P_{\Vec{\alpha}}[\cdot]P_{\Vec{\beta}}^\dag+\exp(-\ii\phi^{l,\Vec{j}}_{\Vec{\alpha},\Vec{\beta}}) P_{\Vec{\beta}}[\cdot]P_{\Vec{\alpha}}^\dag\Big)\Bigg)\notag\\
    =&\mu(\mc N_{l,k})\sum_{\Vec{j}\in\mathbb{F}_3^k}\Pr(\Vec{j}\,|\,l,k)\sum_{\Vec{\alpha},\Vec{\beta}}\Pr(\Vec{\alpha},\Vec{\beta}\,|\,l,\Vec{j},k)\mc V_{(l,\Vec{j}),\Vec{\alpha},\Vec{\beta}},
\end{align}
where we recruit $\Vec{j}\coloneqq(j_1,\cdots,j_k)$, $\Vec{\alpha}\coloneqq(\alpha_{1},\cdots,\alpha_{k})$, $P_{\Vec{\alpha}}\coloneqq P_{\alpha_{1}}\cdots P_{\alpha_{k}}$, $\Vec{\beta}\coloneqq(\beta_{1},\cdots,\beta_{k})$, and $P_{\Vec{\beta}}\coloneqq P_{\beta_{1}}\cdots P_{\beta_{k}}$.
The overall phase $\phi^{l,\Vec{j}}_{\Vec{\alpha},\Vec{\beta}}$ is the summation of the phases from all $\{\chi^{(l,j_i)}_{\alpha_{i}\beta_{i}}\}_i$ coefficients and the phase from the sign of $a_{\Vec{j}}^{(k)}$. 
We use $\mc V_{(l,\Vec{j}),\Vec{\alpha},\Vec{\beta}}$ to denote corresponding Pauli conjugates.
The 1-norm $\mu(\mc N_{l,k})\coloneqq\sum_{\Vec{j}\in\mathbb{F}_3^k}\abs{a^{(k)}_{\Vec{j}}}\prod_{i=1}^k\abs{\chi^{(l,j_i)}}$.
The probability satisfies $\Pr(\Vec{j}\,|\,l,k)\coloneqq\abs{a^{(k)}_{\Vec{j}}}(\prod_{i=1}^k\abs{\chi^{(l,j_i)}})/\mu(\mc N_{l,k})$.
The conditional probability $\Pr(\Vec{\alpha},\Vec{\beta}\,|\,l,\Vec{j},k)$ satisfies
\begin{align}
    \Pr(\Vec{\alpha},\Vec{\beta}\,|\,l,\Vec{j},k)=\prod_{i=1}^k\Pr(\alpha_{j_i},\beta_{j_i}\,|\,l,j_i)=\prod_{i=1}^k\frac{\abs{\chi_{\alpha_{j_i}\beta_{j_i}}^{(l,j_i)}}}{\abs{\chi^{(l,j_i)}}}.
\end{align}

To derive the overall compensation, we can simply take the summation over all $k$-order formulas.
Nevertheless, we always need to truncate the series of $\mc N_{l}(\tau)$ to a fixed order $K_l$ for practice, which we denote by $\tilde{N}_{l}(\tau)$.
According to the recursive solutions of $a$, we know that $\mc N_{l,0}=\mc I$ and $\mc N_{l,1}=0$.
Therefore, we get the expansion of $\tilde{N}_{l}(\tau)$ as
\begin{align}\label{eq:LCS_N}
    \tilde{\mc N}_{l}(\tau)\coloneqq\sum_{k=0}^{K_l}\tau^k\mc N_{l,k}
    =\mu(\tilde{\mc N}_{l}(\tau))\left(\Pr(0\,|\,l)\mc I+\sum_{k=2}^{K_l}\Pr(k\,|\,l)\frac{\mc N_{l,k}}{\mu(\mc N_{l,k})}\right),
\end{align}
where the 1-norm is
\begin{align}\label{eq:1-norm_N}
    \mu(\tilde{\mc N}_{l}(\tau))\coloneqq&1+\sum_{k=2}^{K_l}\tau^k\mu(\mc N_{l,k})=1+\sum_{k=2}^{K_l}\sum_{j_1\cdots j_k=1}^3|a_{j_1\cdots j_k}^{(k)}|\left(\prod_{i=1}^k\chi^{(l,j_i)}\right)\tau^k\leq 1+\sum_{k=2}^{K_l}\lambda^{k-2}a_2\|D_l\|_1^{2k}\tau^k\notag\\
    \leq&1+\frac{a_2\|D_l\|_1^4\tau^2}{1-\lambda\|D_l\|_1^2\tau}\leq 1+2a_2\|D_l\|_1^4\tau^2\leq\mathrm{e}^{2a_2\|D_l\|_1^4\tau^2},
\end{align}
where the first inequality comes from Lemma~\ref{lm:geom} and Eq.~\eqref{eq:1normL1}, and the third comes from the constraint of a small $\tau$.
The probabilities thus satisfy $\Pr(0\,|\,l)=1/\mu(\tilde{\mc N}_{l,\tau})$ and $\Pr(k\,|\,l)=\tau^k\mu(\mc N_{l,k})/\mu(\tilde{\mc N}_{l}(\tau))$.

To quantify the bias of this formula toward the ideal $\mc N_{l}(\tau)$, we use the diamond-norm distance:
\begin{align}
    \left\|\mc N_{l}(\tau)-\tilde{\mc N}_{l}(\tau) \right\|_\diamond&=\left\|\sum_{k=K_l+1}^\infty \tau^k\mc N_{l,k}\right\|_\diamond\leq\sum_{k=K_l+1}^\infty\tau^k\|\mc N_{l,k}\|_\diamond\notag\\
    &\leq \sum_{k=K_l+1}^\infty\tau^k\mu(\mc N_{l,k})\sum_{\Vec{j}\in\mathbb{F}_3^k}\Pr(\Vec{j}\,|\,l,k)\sum_{\Vec{\alpha},\Vec{\beta}}\Pr(\Vec{\alpha},\Vec{\beta}\,|\,l,\Vec{j},k)\|\mc V_{(l,\Vec{j})),\Vec{\alpha},\Vec{\beta}}\|_\diamond\notag\\
    &=\sum_{k=K_l+1}^\infty\tau^k\mu(\mc N_{l,k})\leq\frac{a_2\lambda^{K_l-1}\|D_l\|_1^{2K_l+2}\tau^{K_l+1}}{1-\lambda\|D_l\|_1^2\tau}\leq2a_2\lambda^{K_l-1}\|D_l\|_1^{2K_l+2}\tau^{K_l+1}.
\end{align}

To implement this $\tilde{\mc N}_l(\tau)$, we can use the circuit in Sec.~\ref{sec:LCS}.
Given that $\tilde{\mc N}_l(\tau)$ is an LCS formula with the Pauli-conjugate basis, the circuit requires one ancilla qubit and two control-Pauli gates.
Moreover, the Pauli terms in the LCS can achieve the locality of $\order{K_N}$, which implies $\order{K_N}$ elementary gates. 
\end{proof}

\subsection{High-Accuracy Hamiltonian Simulations}\label{sec:H}
Here, we focus on simulating the Hamiltonian evolution.
As illustrated previously, we will use a two-stage approach to achieve this simulation.
For the coarse-grained simulation, we choose the Trotter-Suzuki method, which stands as the most accessible approach to near-term quantum devices.
For simplicity, we denote the $p$th-order Trotter formula unitary by $S_p(\tau)$.

Even though the coarse-grained stage is accessible for practical implementation, it will incur some systematic errors as a trade-off.
Specifically, we can represent the corresponding error of $S_p(\tau)$ in a multiplicative manner:
\begin{gather}
    V_p(\tau)S_p(\tau)=\mathrm{e}^{-\ii H\tau}\Rightarrow V_p(\tau)=\mathrm{e}^{-\ii H\tau}S_p(\tau)^\dag.
\end{gather}
In~\cite{zeng2022simple}, the authors explicitly analyzed the unitary compensation $V_p(\tau)$.
They decompose the compensation on the Pauli-rotation basis.

We follow this approach but rephrase it in the superoperator representation.
In the following proposition, we construct a Pauli-rotation-conjugate LCS formula for comepnsating the Trotter error. 
\begin{proposition}\label{prop:H_simulation}
    Consider a Hamiltonian $H$ and a short time $\tau<1/(2\|H\|_1)$.
    The compensation operation $\mc V_p(\tau)\coloneqq V_p(\tau)[\cdot]V_p(\tau)^\dag$ can be constructed as a $(\mu,\epsilon)$-LCS formula with the Pauli-rotation-conjugate basis such that     \begin{align*}
        \mu\leq\mathrm{e}^{2c(2\|H\|_1\tau)^{2p+2}},\hspace{0.5em}\text{and}\hspace{0.5em}
        \epsilon\leq2\left(\frac{2\mathrm{e}\|H\|_1\tau}{K_V+1}\right)^{K_V+1}+\left(\frac{2\mathrm{e}\|H\|_1\tau}{K_V+1}\right)^{2K_V+2},
    \end{align*}
    This LCS can be realized using $\order{K_V}$ elementary gates with one ancilla qubit.
\end{proposition}
\begin{proof}
    According to the definition of the LCU formula and Proposition 8 in~\cite{zeng2022simple}, we know that the unitary compensation $V_p(\tau)$ can be approximated by truncating up to its $K_V$th-order terms $\tilde{V}_{p}(\tau)$.
\begin{gather}
    \epsilon_V\coloneqq\|\tilde{V}_{p}(\tau)-\mathrm{e}^{-\ii H\tau}S_p(\tau)^\dag\|\leq\left(\frac{2\mathrm{e}\|H\|_1\tau}{K_V+1}\right)^{K_V+1}.
\end{gather}
Moreover, we can represent compensation as a convex combination of:
    \begin{align}
        \tilde{V}_{p}(\tau)=&\mu_V\sum_{i}\Pr(i)U_i,
    \end{align}
    where $\{\Pr(i)\}$ are probabilities and $\{U_i\}$ are Pauli-rotation unitaries. The normalization factor satisfies
        \begin{align}
        \mu_V\leq\mathrm{e}^{c(2\|H\|_1\tau)^{2p+2}}.
    \end{align}
    Therefore, we can represent the corresponding superoperator, $\tilde{\mc V}_p(\tau)\coloneqq\tilde{V}_{p}(\tau)[\cdot]\tilde{V}_{p}(\tau)^\dag$, as a convex combination with a normalization factor,
    \begin{align}\label{eq:V_LCS}
        \tilde{\mc V}_p(\tau)=\mu_V^2\sum_{i,j}\Pr(i)\Pr(j)U_i[\cdot]U_j^\dag=\mu_V^2\sum_{i,j}\Pr(i)\Pr(j)\cdot\frac{1}{2}(U_i[\cdot]U_j^\dag+U_j[\cdot]U_i^\dag)
    \end{align}
    This generates an LCS formula toward the $\mc V_p(\tau)$ with the normalization factor $\mu=\mu_V^2$. 
    The accuracy of this formula is quantified via the diamond norm as
    \begin{align}
        \left\|\tilde{\mc V}_p(\tau)-\mc V_p(\tau)\right\|_\diamond=&\left\|\tilde{V}_{p}(\tau)[\cdot]\tilde{V}_{p}(\tau)^\dag-\mathrm{e}^{-\ii H\tau}S_{p}(\tau)^\dag[\cdot]S_{p}(\tau)\mathrm{e}^{\ii H\tau}\right\|_\diamond\notag\\
    \leq&2\epsilon_V+\epsilon_V^2\leq2\left(\frac{2\mathrm{e}\|H\|_1\tau}{K_V+1}\right)^{K_V+1}+\left(\frac{2\mathrm{e}\|H\|_1\tau}{K_V+1}\right)^{2K_V+2}.
    \end{align}

    Since the LCS formula in Eq.~\eqref{eq:V_LCS} is a decomposition on the Pauli-rotation conjugates, the compensation $\tilde{\mc V}_p(\tau)$ can be implemented by two controlled unitaries with one ancilla qubit.
    Further compiling these two unitary gates results in $\order{K_V}$ gates due to the $\order{K_V}$ locality of the corresponding Pauli~\cite{mansky2023decomposition}.
\end{proof}

In our Lindblad simulation, we frequently use Hamiltonian evolutions as subroutines.
This compensation facilitating the high-accuracy Hamiltonian simulation also helps to implement our Lindblad simulation.

The most straightforward application is regarding the Hamiltonian simulation $\mathrm{e}^{\mc H\tau}$ in Eq.~\eqref{eq:Lind_de_app}, where we employ a first-order Trotter $\mc S_1(\tau)$ as the coarse-grained simulation and use $\tilde{\mc V}_1(\tau)$ to refine the simulation.
This design allows a high-accuracy approximation of $\mathrm{e}^{\mc H\tau}$.
By regarding $\mc S_1(\tau)$ as a $(1,0)$-LCS formula, we can get the following statement:
\begin{corollary}\label{co:F}
    Consider a $q$-sparse Hamiltonian $H$ and a short time $\tau<1/(2\|H\|_1)$.
    The concatenated superoperator, $\tilde{\mc F}(\tau)\coloneqq\tilde{\mc V}_1(\tau)\circ\mc S_1(\tau)$, is a $(\mu,\epsilon)$-LCS formula of $\mathrm{e}^{\mc H\tau}$ such that
    \begin{align*}
        \mu\leq\mathrm{e}^{2c(2\|H\|_1\tau)^{4}},\hspace{0.5em}\text{and}\hspace{0.5em}
        \epsilon\leq2\left(\frac{2\mathrm{e}\|H\|_1\tau}{K_F+1}\right)^{K_F+1}+\left(\frac{2\mathrm{e}\|H\|_1\tau}{K_F+1}\right)^{2K_F+2},
    \end{align*}
    where $K_F$ is the truncation order, and $c\leq3.5$ is a constant.
    This LCS can be realized using $\order{q+K_F}$ elementary gates with one ancilla qubit.
\end{corollary}

Another scenario where we need Hamiltonian simulation emerges in dissipative parts.
In that case, every dissipative term $\mc J_l(\tau)$ in Eq.~\eqref{eq:J_def_app} requires an $(n+1)$-qubit Hamiltonian simulation $\mathrm{e}^{-\ii J_l\sqrt{\tau}}[\cdot]\mathrm{e}^{\ii J_l\sqrt{\tau}}$.
Similarly, our method uses a first-order Trotter approach to fulfill the coarse-grained simulation followed by the compensation term.
Therefore, we can get the approximation:
\begin{gather}\label{eq:tilde_J}
    \tilde{\mc J}_{l}(\tau)[\rho]\coloneqq\Tr_a[\tilde{\mc V}_{1,l}(\sqrt{\tau})\circ\mc S_{1,l}(\sqrt{\tau}) [\ket{0}\bra{0}_a\otimes \rho]]
\end{gather}
This is also a valid LCS formula of $\mc J_l(\tau)$:
\begin{corollary}\label{co:J}
    Consider a short time $\tau<1/(16\|D_l\|_1^2)$ and $q$-sparse jump operator $D_l$.
    The superoperator $\tilde{\mc J}_l(\tau)$ in Eq.~\eqref{eq:tilde_J} is a $(\mu,\epsilon)$-LCS formula of $\mc J_l(\tau)$ such that
    \begin{align*}
        \mu\leq\mathrm{e}^{2c(16\|D_l\|^2_1\tau)^{2}},\hspace{0.5em}\text{and}\hspace{0.5em}
    \epsilon\leq2\left(\frac{4\mathrm{e}\|D_l\|_1\sqrt{\tau}}{K_J+1}\right)^{K_J+1}+\left(\frac{4\mathrm{e}\|D_l\|_1\sqrt{\tau}}{K_J+1}\right)^{2K_J+2},
    \end{align*}
    where $K_J$ is the truncation order, and $c\leq3.5$ is a constant.
    This LCS can be realized using $\order{q+K_J}$ elementary gates with two ancilla qubit.
\end{corollary}

\subsection{Simulating Time-independent Evolutions}\label{sec:independent}

Based on the previous discussion, we have already elaborated on the construction of all parts in the simulation as well as the compensation for the Lie-Trotter formula.
Therefore, we can approximate the evolution in each time step by
\begin{equation}\label{eq:Lind_decomp_app}
\begin{aligned}
		\mathrm{e}^{\mc L\tau}=\mc M(\tau) \circ \overleftarrow{\prod_{l=1}^m}\left(\mc N_{l}(\tau) \circ \mc J_{l}(\tau)\right)\circ \mathrm{e}^{\mc H\tau}\approx\tilde{\mc M}(\tau) \circ \overleftarrow{\prod_{l=1}^m}\left(\tilde{\mc N}_{l}(\tau)\circ\tilde{\mc J}_{l}(\tau)\right)\circ \tilde{\mc F}(\tau).
\end{aligned}
\end{equation}
Consequently, we can organize the short-time simulation as Eq.~\eqref{eq:Lind_decomp_app} and repeat the process for $ T/\tau$ times to get the desired evolution with time $T$.
For simplicity, we assume that $T/\tau$ is an integer, but we can still deal with general cases by adjusting the time in the last epoch of the simulation.
We summarize the whole simulation in the Alg.~\ref{alg:simulation}, where we utilize comparable orders of truncations $K$ (or $2K$) for each LCS formula.

\begin{algorithm}[t]
  \caption{Time-independent Simulation}\label{alg:simulation}
  \KwIn{State $\rho_0$; Observable $O$; Hamiltonian $H$ and Jump operators $\{D_l\}_{l=1}^m$; Time $T$ and step length $\tau$; Round number $N$; Truncated order $K$ }
  \KwOut{The estimation of $\Tr(\mathrm{e}^{\mc LT}[\rho_0]O)$, $\tilde{\langle O\rangle}$}
  Calculate $\{\chi^{(l)}_{\alpha,\beta}\}_{l=0}^m$ as in Eq.~\eqref{eq:Lind_chi}, $\{\chi^{(l,i)}_{\alpha,\beta}\}_{l=1,i=1}^{m,3}$ as in Eq.~\eqref{eq:L1_chi}, and $\{a^{(k)}_{j_1,\cdots,j_k}\}$ as in Eq.~\eqref{eq:recur_def}\;
  Calculate 1-norms of the $K$th-order truncated LCS formulas $\{\mu(\tilde{\mc N}_{l,\tau})\}_l$ and $\mu(\tilde{\mc M}_{\tau})$ from Eq.~\eqref{eq:1-norm_N} and Eq.~\eqref{eq:1-norm_M}\;
  $\text{Ans}\leftarrow 0$\;
  \For(\tcp*[f]{$N$ rounds}){$i=1,\cdots,N$}{
    \For(\tcp*[f]{$T/\tau$ steps}){$j=1,\cdots,T/\tau$}{
      Apply the first-order Trotter of $\mathrm{e}^{\mc H\tau}$ to the state\;
      Recall Alg.~\ref{alg:random_sample} to apply the $K$th-order truncated LCS formula $\tilde{\mc V}_1(\tau)$ in $\tilde{\mc F}(\tau)$\;
      \For{$l=1,\cdots,m$}{
        Apply the first-order Trotter of $\mathrm{e}^{\mc J_l\sqrt{\tau}}$ to the state and ancilla\;
        Recall Alg.~\ref{alg:random_sample} to apply the $2K$th-order truncated LCS formula $\tilde{\mc V}_{1,l}(\tau)$ in $\tilde{\mc J}_l(\tau)$ to the state and ancilla\;
        Recall Alg.~\ref{alg:random_sample} to apply the $K$th-order truncated LCS formula $\tilde{\mc N}_{l}(\tau)$ as defined in Eq.~\eqref{eq:LCS_N} to the state
      }
      Recall Alg.~\ref{alg:random_sample} to apply the $K$th-order truncated LCS formula $\tilde{\mc M}(\tau)$ as defined in Eq.~\eqref{eq:LCS_M} to the state\;
    }
    Measure the final state $\rho$ on $O$\;
    $\text{Ans}\,+=(\mu(\tilde{\mc M}_{\tau})\prod_{l=1}^m\mu(\tilde{\mc N}_{l,\tau}))^{T/\tau}\times\Tr(\rho O)$
  }
  \Return $\text{Ans}/N$ 
\end{algorithm}

\begin{theorem}\label{thm:time-indep}
    Given a Lindbladian $\mc L$ with $q$-sparse local Hamiltonian $\mc H$ and jump operators $\{D_l\}_{l=1}^m$, a state $\rho_0$, and time $T$, execute Alg.~\ref{alg:simulation} with $\tau=\Theta\left(\frac{1}{\|\mc L\|_\p^2T}\right)$, $K=\Theta\left(\frac{\ln(1/\varepsilon)}{\ln(\|\|\mc L\|_\p T)}\right)$, and $N=\Theta\left(\frac{1}{\varepsilon^2}\log(1/\delta)\right)$ can output an estimation $\tilde{\langle O\rangle}$ satisfying
    \begin{gather}
        \abs{\tilde{\langle O\rangle}-\Tr(\mathrm{e}^{\mc LT}(\rho_0)O)}\leq\varepsilon \|O\|
    \end{gather}
    with probability $1-\delta$. 
    This algorithm requires at most two ancilla qubits at one time with a single-round gate count
    \begin{gather}
        \order{m\|\mc L\|_\p^2T^2\left(q+\frac{\ln(1/\varepsilon)}{\ln(\|\mc L\|_{\rm \p} T)}\right)}.
    \end{gather}
\end{theorem}
\begin{proof}
    
    Recalling Algorithm~\ref{alg:simulation}, the procedure indeed concatenates the LCS formulas for $\mc M_{\tau}$, $\{\mc N_{l,\tau}\}_l$, $\mc J_{l,\tau}$, and the Hamiltonian simulation $\mathrm{e}^{\mc H\tau}$.
    According to Lemma~\ref{lm:concate} and Eq.~\eqref{eq:Lind_decomp_app}, the resulting formula targets the short-time evolution.
    Repeating this for $T/\tau$ times returns the long-time evolution.

    As for the parameters of the resulting formula, we can multiply over all LCS formulas to determine the total 1-norm. 
    Given that $\tau=\Theta{\frac{1}{\|\mc L\|_\p^2T}}$ satisfies the conditions in Propositions~\ref{prop:trotter},~\ref{prop:N_LCS}, and Corollary~\ref{co:F} and~\ref{co:J}, we can bound the 1-norm of the overall formula by:
    \begin{align}
        \mu\coloneqq&\left(\mu(\tilde{\mc M}(\tau))\left(\prod_{l=1}^m\mu(\tilde{\mc N}_{l}(\tau))\mu(\tilde{\mc J}_{l}(\tau))\right)\mu(\tilde{\mc F}(\tau))\right)^{T/\tau}\notag\\
        \leq&\exp(\left(16(\|H\|_1+\sum_{l=1}^m\|D_l\|_1^2)^2+\sum_{l=1}^m(2a_2\|D_l\|_1^4+2c(4\|D_l\|_1)^{4})+2c(2\|H\|_1)^{4}\tau^{2}\right)\tau T)=\order{1},
    \end{align}
    where $c\leq3.5$.

    It is more complicated to calculate the overall bias of this formula. 
    We start by considering the bias of each single-term dissipative evolution.
    We take the $l$th term, $\tilde{\mc N}_{l,\tau}\tilde{\mc J}_{l,\tau}$, as an example.
    Note that $\mc J_{l,\tau}$ is CPTP, and we have
    \begin{align}
        \epsilon_{D_l}\leq\left(2\left(\frac{4\mathrm{e}\|D_l\|_1\sqrt{\tau}}{2K+1}\right)^{2K+1}+\left(\frac{4\mathrm{e}\|D_l\|_1\sqrt{\tau}}{2K+1}\right)^{4K+2}\right)\mathrm{e}^{2a_2\|D_l\|_1^4\tau^2}+2a_2\lambda^{K-1}\|D_l\|_1^{2K+2}\tau^{K+1}=\order{\|D_l\|_1^{2K+1}\tau^{K+1/2}},
    \end{align}
    where we have used $\tau=\Theta{\frac{1}{\|\mc L\|_\p^2T}}\ll1$ and $\|D_l\|_1^2\tau\ll1$.
    Since all of $\{\mc N_{l}(\tau)\mc J_{l}(\tau)\}$ and $\mathrm{e}^{\mc H\tau}$ are CPTP maps, we can further derive the bias 
    \begin{align}
        \epsilon_{H+D}\leq&(1+\epsilon_H)\prod_{l=1}^m(1+\epsilon_{D_l})-1\leq\exp(\epsilon_H+\sum_{l=1}^m\epsilon_{D_l})-1
    \end{align}
    Recall the bounds of $\epsilon_H$ and $\epsilon_M$ from Proposition~\ref{prop:H_simulation} and~\ref{prop:trotter}.
    Noting $\|\mc L\|_\p\tau\ll1$, we can bound the bias of a short-time implementation as
    \begin{align}\label{eq:step_error}
        \epsilon_{\tau}\leq&\epsilon_{H+D}\mu(\tilde{\mc M}_{\tau})+\epsilon_M
        \leq \order{\|\mc L\|_\p^K\tau^K},
    \end{align}
    where we used that $\mathrm{e}^{\mc H\tau}\circ\prod_{l=1}^m\mathrm{e}^{\mc D_l\tau}$ is a CPTP map and that $\mu(\tilde{\mc M}(\tau))\leq\mu=\order{1}$.
    Since this short-time implementation is a formula targeting a CPTP map, the bias of the overall LCS formula can be asymptotically bounded by: 
    \begin{align}
        \epsilon_T\leq (1+\epsilon_{\tau})^{T/\tau}-1
        \leq\exp(T\epsilon_\tau/\tau)-1=\order{\|\mc L\|_\p^K\tau^{K-1}T},
    \end{align}
    where we have used the fact that $\|\mc L\|_\p^K\tau^{K-1}T\ll1$.
    By choosing $K=\order{\frac{\ln(1/\varepsilon)}{\ln(\|\mc L\|_\p T)}}$ with proper constants, we can bound the bias as $\epsilon_T\leq\varepsilon/2$.
    
    The bias analyzed above can be regarded as the distance between the expectation of the output from the LCS formula. 
    Nevertheless, we still need to account for the statistical fluctuation.
    The single-shot measured value $\Tr(\rho O)$ is in the range of $[-\mu\|O\|,\mu\|O\|]$ and is sampled for $N$ rounds independently.
    Therefore, we can bound the sampling fluctuation based on the Hoeffding bound as
    \begin{align}
        \Pr(\abs{\tilde{\langle O\rangle}-\Bar{\langle O\rangle}}\geq\epsilon_s\|O\|)\leq2\exp(-\frac{N\epsilon^2_s}{2\mu^2}).
    \end{align}
    By setting $N=8\mu^2\log(2/\delta)/\varepsilon^2=\order{\log(1/\delta)/\varepsilon^2}$ and $\epsilon_s=\varepsilon/2$, this guarantees the overall error to be
    \begin{gather}
        \abs{\tilde{\langle O\rangle}-\Tr(\mathrm{e}^{(\mc H+\mc D)T}[\rho_0]O)}\leq\abs{\tilde{\langle O\rangle}-\Bar{\langle O\rangle}}+\abs{\Bar{\langle O\rangle}-\Tr(\mathrm{e}^{(\mc H+\mc D)T}[\rho_0]O)}\leq(\epsilon_s+\epsilon_T)\|O\|\leq \varepsilon\|O\|,
    \end{gather}
    with probability $1-\delta$.

    To count the gate overheads, we must notice that the whole algorithm implements the short-time formula for $T/\tau=\order{\|\mc L\|_\p^2T^2}$ times.
    By enumerating the four subroutines of LCS formulas in Alg.~\ref{alg:simulation}, we know that the short-time simulation can be realized using $\order{m(q+K)}$ elementary gates and at most two ancilla qubits.
    This comes to the overall gate count as
    \begin{gather}
        \order{m\|\mc L\|_\p^2T^2\left(q+\frac{\ln(1/\varepsilon)}{\ln(\|\mc L\|_\p T)}\right)}.
    \end{gather}
\end{proof}

\section{Simulating Time-dependent Evolutions}\label{sec:depend}
When referring to the time-dependent Lindblad equation, we consider a bounded continuous Superoperator-valued function $\mc L(t)$.
The evolution from $t=0$ to $T$ can be written in a time-ordered exponential and decomposed via \emph{Dyson series}.
\begin{gather}
    \mc T\mathrm{e}^{\int_{0}^T\mc L(t)dt}=\mc I+\sum_{k=1}^\infty\frac{1}{k!}\int_{0}^Tdt_1\int_0^{T}dt_2\cdots\int_{0}^{T}dt_k\cdot\mc T\overrightarrow{\prod_{i=1}^k}\mc L(t_i),
\end{gather}
where $\mc T$ represents time-ordered permutation.

To simulate this time-dependent process, we decompose the overall evolution into several small steps with length $\tau\ll1$.
Without loss of generality, we assume the time step we considered starts from $t=t_0$ and ends at $t=t_0+\tau$.
We can further partition this step into $M$ slices, each of which has a length of $\tau/M$.
For each slice, we take its middle superoperator as the representative, denoted by
\begin{align}\label{eq:slice}
    \mc L_j\coloneqq\mc L\left(\frac{(j-1/2)\tau}{M}+t_0\right)=\sum_{\alpha,\beta}\chi^{(\mc L_j)}_{\alpha\beta}P_\alpha(\cdot)P_\beta^\dag\ \ \ \forall\,j\in[M].
\end{align}
This piece-wise constant Lindbladian, denoted by $\mc L_M(t)$, is close to the original one,
\begin{gather}
    \|\mc L_M(t)-\mc L(t)\|_\diamond\leq\max_t\left\|\dot{\mc L}(t)\frac{\tau}{2M}\right\|_\diamond=\frac{\tau\|\dot{\mc L}\|_\diamond}{2M},\ \forall\,t\in[t_0,t_0+\tau],
\end{gather}
where $\|\dot{\mc L}\|_\diamond\coloneqq\max_t\|\dot{\mc L}(t)\|_\diamond$.
Given $\|\mc L\|_\p\tau\ll1$, this new Lindbladian approximates the short-time original evolution.
\begin{align}\label{eq:discrete_error_app}
\|\mc T\mathrm{e}^{\int_{t_0}^{t_0+\tau}\mc L_M(t)dt}-\mc T\mathrm{e}^{\int_{t_0}^{t_0+\tau}\mc L(t)dt}\|_\diamond=&\left\|\sum_{k=0}^\infty\frac{1}{k!}\int_{t_0}^{t_0+\tau} dt_1\int_{t_0}^{t_0+\tau}dt_2\cdots\int_{t_0}^{t_0+\tau}dt_k\cdot\left(\mc T\overrightarrow{\prod_{i=1}^k}\mc L_M(t_i)-\mc T\overrightarrow{\prod_{i=1}^k}\mc L(t_i)\right)\right\|_\diamond\notag\\
\leq&\sum_{k=1}^\infty\frac{1}{k!}\int_{t_0}^{t_0+\tau} dt_1\int_{t_0}^{t_0+\tau}dt_2\cdots\int_{t_0}^{t_0+\tau}dt_k\left\|\mc T\overrightarrow{\prod_{i=1}^k}\mc L_M(t_i)-\mc T\overrightarrow{\prod_{i=1}^k}\mc L(t_i)\right\|_\diamond\notag\\
\leq&\sum_{k=1}^\infty\frac{k\tau^k}{k!}\frac{\tau\|\dot{\mc L}\|_\diamond}{2M}\|\mc L\|_\diamond^{k-1}\leq\frac{\tau^2\|\dot{\mc L}\|_\diamond}{2M}\mathrm{e}^{\|\mc L\|_\p\tau}=\order{\frac{\tau^2\|\dot{\mc L}\|_\diamond}{M}}.
\end{align}

Nevertheless, we can further compromise by simulating the time-independent Lindbladian, $\mathrm{e}^{\sum_{j=1}^M\mc L_j\tau/M}$, to approximate the evolution of $\mc L_M$.
To improve the accuracy of this approximation, we further introduce a compensation part $\mc W(\tau)$.
We can expand this by the Dyson series:
\begin{align}\label{eq:W_def_app}
    \mc W(\tau)\coloneqq\mc T\mathrm{e}^{\int_{t_0}^{t_0+\tau}\mc L_M(t)dt}\circ\mathrm{e}^{-\sum_{j=1}^M\mc L_j\tau/M}=\sum_{s_1=0}^\infty\frac{\tau^{s_1}}{M^{s_1}s_1!}\mc T\left(\overrightarrow{\prod_{i=1}^{s_1}}\sum_{j_i=1}^M\mc L_{j_i}\right)\circ\mathrm{e}^{-\sum_{j=1}^M\mc L_j\tau/M}=\sum_{k=0}^\infty\tau^k\mc W_{k},
\end{align}
where $\mc T$ imposes the necessary reordering of labels $\{j_i\}_i$ to make sure the product is in the time order.
We represent it as an LCS formula to implement this compensation.
\begin{proposition}\label{prop:time_dep}
    The compensation $\mc W(\tau)$ in Eq.~\eqref{eq:W_def_app} can be represented as a $(\mu,\epsilon)$-LCS formula on Pauli-conjugate basis such that 
    \begin{gather*}
        \mu\leq\mathrm{e}^{\mathrm{e}(2\|\mc L\|\tau)^3},\text{ and }
        \epsilon\leq\left(\frac{2\mathrm{e}\|\mc L\|_\p\tau}{K_W+1}\right)^{K_W+1},
    \end{gather*}
    where $K_W$ is the truncation order.
    The gate count for this formula is $\order{K_W}$ using one ancilla qubit.
\end{proposition}

\begin{proof}
It can be verified that $\mc W_0=\mc I$ and $\mc W_1=\mc W_2=0$ in Eq.~\eqref{eq:W_def_app}.
Thus, we only decompose the higher-order component $\mc W_{k}$ as
\begin{align}
   \mc W_{k}\coloneqq&\sum_{s_1+s_2=k}\frac{(-1)^{s_2}}{M^{k}s_1!s_2!}\mc T\left(\overrightarrow{\prod_{i=1}^{s_1}}\sum_{j_i=1}^M\mc L_{j_i}\right)\circ\left(\sum_{j=1}^M\mc L_j\right)^{\circ s_2}\notag\\
   =&\sum_{s_1+s_2=k}\frac{(-1)^{s_2}}{M^{k}s_1!s_2!}\mc T\left(\sum_{j=1}^M\sum_{\alpha,\beta}\chi^{(\mc L_{j})}_{\alpha\beta}P_\alpha(\cdot)P_\beta^\dag\right)^{\circ s_1}\circ\left(\sum_{j=1}^M\sum_{\alpha,\beta}\chi^{(\mc L_j)}_{\alpha\beta}P_\alpha(\cdot)P_\beta^\dag\right)^{\circ s_2}\notag\\
   =&\sum_{s_1+s_2=k}\frac{(-1)^{s_2}}{M^{k}s_1!s_2!}\mc T\left(\sum_{j=1}^M\abs{\chi^{(\mc L_j)}}\sum_{\alpha,\beta}\frac{\chi^{(\mc L_{j})}_{\alpha\beta}}{\abs{\chi^{(\mc L_j)}}}P_\alpha(\cdot)P_\beta^\dag\right)^{\circ s_1}\circ\left(\sum_{j=1}^M\abs{\chi^{(\mc L_j)}}\sum_{\alpha,\beta}\frac{\chi^{(\mc L_{j})}_{\alpha\beta}}{\abs{\chi^{(\mc L_j)}}}P_\alpha(\cdot)P_\beta^\dag\right)^{\circ s_2}\notag\\
   =&\sum_{s_1+s_2=k}\frac{(-1)^{s_2}\left(\sum_{j=1}^M\abs{\chi^{(\mc L_j)}}\right)^{k}}{M^{k}s_1!s_2!}\sum_{\Vec{j}\in[M]^k}\Pr(\Vec{j}\,|\,s_1,s_2,k)\mc T\left(\overrightarrow{\prod_{i=1}^{s_1}}\sum_{\alpha,\beta}\frac{\chi^{(\mc L_{j_i})}_{\alpha\beta}}{\abs{\chi^{(\mc L_{j_i})}}}P_\alpha(\cdot)P_\beta^\dag\right)\circ\left(\overrightarrow{\prod_{i=s_1+1}^{k}}\sum_{\alpha,\beta}\frac{\chi^{(\mc L_{j_i})}_{\alpha\beta}}{\abs{\chi^{(\mc L_{j_i})}}}P_\alpha(\cdot)P_\beta^\dag\right)\notag\\
   =&\mu(\mc W_k)\sum_{\substack{s_1+s_2=k\\\Vec{j}\in[M]^k}}\Pr(\Vec{j},s_1,s_2\,|\,k)(-1)^{s_2}\sum_{\Vec{\alpha},\Vec{\beta}}\left(\prod_{i=1}^{s_1}\frac{\chi^{(\mc L_{j_{\mc T^{-1}(i)}})}_{\alpha_i\beta_i}}{\abs{\chi^{(\mc L_{j_{\mc T^{-1}(i)}})}}}\right)\left(\prod_{i=s_1+1}^{k}\sum_{\alpha,\beta}\frac{\chi^{(\mc L_{j_i})}_{\alpha_i\beta_i}}{\abs{\chi^{(\mc L_{j_i})}}}\right)P_{\Vec{\alpha}}(\cdot)P_{\Vec{\beta}}^\dag\notag\\
   =&\mu(\mc W_k)\sum_{\substack{s_1+s_2=k\\\Vec{j}\in[M]^k}}\Pr(\Vec{j},s_1,s_2\,|\,k)\sum_{\Vec{\alpha},\Vec{\beta}}\left(\prod_{i=1}^{s_1}\Pr(\alpha_i,\beta_i\,|\,\mc L_{j_{\mc T^{-1}(i)}})\right)\left(\prod_{i=s_1+1}^{k}\sum_{\alpha,\beta}\Pr(\alpha_i,\beta_i\,|\,\mc L_{j_i})\right)\mc V_{(\Vec{j},s_1,s_2),\Vec{\alpha},\Vec{\beta}},
\end{align}
where we similarly use $\Vec{j}\coloneqq(j_1,\cdots,j_k)$, $\Vec{\alpha}\coloneqq(\alpha_{1},\cdots,\alpha_{k})$, $P_{\Vec{\alpha}}\coloneqq P_{\alpha_{1}}\cdots P_{\alpha_{k}}$, $\Vec{\beta}\coloneqq(\beta_{1},\cdots,\beta_{k})$, and $P_{\Vec{\beta}}\coloneqq P_{\beta_{1}}\cdots P_{\beta_{k}}$.
The $\mc T$ here is a permutation of the labels of $\{j_i\}$ to ensure the time ordering.
The map $\mc V_{(\Vec{j},s_1,s_2),\Vec{\alpha},\Vec{\beta}}$ consists of Pauli conjugates from $P_{\Vec{\alpha}}$ and $P_{\Vec{\beta}}$ with phase from all corresponding $\chi$ matrices and the sign of $(-1)^{s_2}$.
The normalization satisfies $\mu(\mc W_k)\coloneqq\frac{(2\sum_{j=1}^M\abs{\chi^{(\mc L_j)}})^k}{M^kk!}$.
The probability of Pauli labels satisfies $\Pr(\alpha,\beta\,|\,\mc L_{j})\coloneqq\frac{\abs{\chi^{(\mc L_{j})}_{\alpha\beta}}}{\abs{\chi^{(\mc L_{j})}}}$.
The probability of all other parameters can be determined independently from multiple distributions
\begin{align}
    \Pr(\Vec{j},s_1,s_2\,|\,k)\coloneqq\Pr(s_1,s_2\,|\,k)\Pr(\Vec{j}\,|\,s_1,s_2,k)=\frac{k!}{2^ks_1!s_2!}\prod_{i=1}^k\frac{\abs{\chi^{(\mc L_{j_i})}}}{\sum_{j=1}^M\abs{\chi^{(\mc L_j)}}}.
\end{align}

Truncate the overall compensation $\mc W(\tau)$ to its $K$th order, and we can get an approximation as
\begin{align}\label{eq:W_LCS}
    \tilde{\mc W}(\tau)\coloneqq&\mc I+\sum_{k=3}^K\tau^k\mc W_k=\mu(\tilde{\mc W}(\tau))\left(\Pr(0)\mc I+\sum_{k=3}^K\Pr(k)\frac{\mc W_k}{\mu(\mc W_k)}\right).
\end{align}
Its normalization factor satisfies
\begin{align}
    \mu(\tilde{\mc W}(\tau))\coloneqq 1+\sum_{k=3}^K\frac{(2\sum_{j=1}^M\abs{\chi^{(\mc L_j)}}\tau)^k}{M^kk!}\leq\mathrm{e}^{2\sum_{j=1}^M\abs{\chi^{(\mc L_j)}}\tau/M}-\sum_{k=1}^2\frac{(2\sum_{j=1}^M\abs{\chi^{(\mc L_j)}}\tau)^k}{M^kk!}
    \leq\mathrm{e}^{\mathrm{e}(2\|\mc L\|\tau)^3},
\end{align}
where we used the maximum of $\chi$ matrix, $\|\mc L\|$.
Thus, the probabilities in Eq.~\eqref{eq:W_LCS} satisfy $\Pr(0)=1/\mu(\tilde{\mc W}(\tau))$ and $\Pr(k)=\mu(\mc W_k)/\mu(\tilde{\mc W}(\tau))$ for $k\geq3$.

The error comes from the truncation and can be bounded by
\begin{align}
    \|\tilde{\mc W}(\tau)-\mc W(\tau)\|_\diamond\leq\sum_{k=K+1}^\infty\tau^k\|\mc W_k\|_\diamond=\sum_{k=K+1}^\infty\frac{(2\sum_{j=1}^M\abs{\chi^{(\mc L_j)}}\tau)^k}{M^kk!}\leq\left(\frac{2\mathrm{e}\|\mc L\|_\p\tau}{K+1}\right)^{K+1}.
\end{align}

To implement this $\tilde{\mc W}(\tau)$, we can use the circuit in Sec.~\ref{sec:LCS}.
Given that $\tilde{\mc W}(\tau)$ is an LCS formula with the Pauli-conjugate basis, the circuit requires one ancilla qubit and two control-Pauli gates.
Moreover, the Pauli terms in the LCS can achieve the locality of $\order{K_W}$, which implies $\order{K_W}$ elementary gates. 

\end{proof}
With this LCS formula, the compensation $\mc W(\tau)$ can be easily implemented via the sampling.
Combining $\mc W(\tau)$ and $\mathrm{e}^{\sum_{j=1}^M\mc L_j\tau/M}$, we can get a good approximation of the time-dependent evolution from $t_0$ to $t_0+\tau$.
By repeating this process of each time step $\tau$ for $T/\tau$ times, we can simulate the overall time-dependent evolution.
We summarize the whole scheme in Algorithm~\ref{alg:simulation2}.

\begin{algorithm}[t]
  \caption{Time-dependent Simulation}\label{alg:simulation2}
  \KwIn{State $\rho_0$; Observable $O$; Time-dependent Lindbladian $\mc L(t)$; Time $T$ and step length $\tau$; Number of Slices $M$; Round number $N$; Truncated order $K$ }
  \KwOut{The estimation of $\Tr(\mc T\mathrm{e}^{\int_0^T\mc L(t)dt}[\rho_0]O)$}
  $\text{Ans}\leftarrow 0$\;
  \For{$n=1,\cdots,N$}{
    \For{$j=1,\cdots,T/\tau$}{
      Calculate $M$ Lindbladians $\{\mc L_{j}\}$ as in Eq.~\eqref{eq:slice} by setting $t_0=(j-1)\tau$\;
      Apply Alg.~\ref{alg:simulation} with truncation $K$ and time length $\tau$ of $\sum_{j=1}^M\mc L_j/M$ with 1-norm to be $\mu_{L,j}$\;
      Apply Alg.~\ref{alg:random_sample} for the $K$th-order truncated convex combination, $\tilde{\mc W}(\tau)$, in Eq.~\eqref{eq:W_LCS} with 1-norm to be $\mu_{W,j}$\;
    }
    Measure the final state $\rho$ on $O$\;
    $\text{Ans}\,+=(\prod_{j=1}^{T/\tau}\mu_{W,j}\mu_{L,j})\times\Tr(\rho O)$
  }
  \Return $\text{Ans}/N$ 
\end{algorithm}

\begin{theorem}\label{thm:time-dep}
    Given a continuous time-dependent Lindbladian $\mc L(t)$ with a bounded first-order derivative $\|\dot{\mc L}\|_\diamond$, a state $\rho_0$ and time $T$, execute Alg.~\ref{alg:simulation2} with $\tau=\Theta\left(\frac{1}{\|\mc L\|^2_\p T}\right)$, $K=\Theta\left(\frac{\ln(1/\varepsilon)}{\ln(\|\mc L\|_\p T)}\right)$, $M=\Theta\left(\frac{\|\dot{\mc L}\|_\diamond}{\varepsilon\|\mc L\|_\p^2}\right)$, and $N=\Theta\left(\frac{1}{\varepsilon^2}\log(1/\delta)\right)$ can output an estimation $\tilde{\langle O\rangle}$ satisfying
    \begin{gather}
        \abs{\tilde{\langle O\rangle}-\Tr(\mc T\mathrm{e}^{\int_{0}^T\mc L(t)dt}[\rho_0]O)}\leq\varepsilon \|O\|
    \end{gather}
    with probability $1-\delta$. 
    This algorithm requires at most two ancilla qubits at one time with a single-round gate count
    \begin{gather}
        \order{m\|\mc L\|_\p^2T^2\left(q+\frac{\ln(1/\varepsilon)}{\ln(\|\mc L\|_\p T)}\right)}.
    \end{gather}
\end{theorem}
\begin{proof}
    According to Theorem~\ref{thm:time-dep} as well as Eq.~\eqref{eq:W_def_app}, we know that each step in Alg.~\ref{alg:simulation2} sequentially implements LCS formulas, which approximate the short-time time-dependent evolutions.
    Note that the target Lindbladian is time-dependent, so operations in Algorithm~\ref{alg:simulation2} must be various for different time steps.
    To avoid the cluttering of notations, we ignore the unnecessary labels inferring the step number in the following derivations.
    
    Without loss of generality, we consider a fixed $j$th step.
    First, we analyze the parameters of this short-time LCS formula.
    According to Lemma~\ref{lm:concate}, the normalization can be calculated by direct multiplication.
    \begin{align}
        \mu_j\coloneqq&\mu_{W,j}\mu_{L,j}=\mu_{W,j}\cdot\mu(\tilde{\mc M}(\tau))\left(\prod_{l=1}^m\mu(\tilde{\mc N}_{l}(\tau))\mu(\tilde{\mc J}_{l}(\tau))\right)\mu(\tilde{\mc F}(\tau))\notag\\
        \leq&\exp(\mathrm{e}(2\|\mc L\|\tau)^3+16(\|H\|_1+\sum_{l=1}^m\|D_l\|_1^2)^2\tau^2+\sum_{l=1}^m(2a_2\|D_l\|_1^4\tau^2+2c(16\|D_l\|^2_1\tau)^{2k+1})+2c(2\|H\|_1\tau)^{4k+2})\notag\\
        \leq&\exp(\order{\|\mc L\|^2\tau^2}),
    \end{align}
    where we used normalization in Propositio~ \ref{prop:N_LCS}, \ref{prop:trotter}, and Corollary~\ref{co:F} and~\ref{co:J}.
    We have also assumed that the step length is as small as $\Theta\left(\frac{1}{\|\mc L\|^2T}\right)$.
    The bias between this $j$th-step LCS formula and the ideal evolution from $(j-1)\tau$ to $j\tau$ is 
    \begin{align}
        \epsilon_j\coloneqq&\left\|\mc T\mathrm{e}^{\int_{(j-1)\tau}^{j\tau}\mc L(s)ds}-\tilde{\mc W}(\tau)\circ\tilde{\mc M}(\tau)\left(\prod_{l=1}^m\tilde{\mc N}_{l}(\tau)\tilde{\mc J}_{l}(\tau)\right)\tilde{\mc F}(\tau)\right\|_\diamond\notag\\
        \leq&\|\mc T\mathrm{e}^{\int_{(t-1)\tau}^{t\tau}\mc L(s)ds}-\mc T\mathrm{e}^{\int_{(t-1)\tau}^{t\tau}\mc L_M(s)ds}\|_\diamond+\|\mc T\mathrm{e}^{\int_{(t-1)\tau}^{t\tau}\mc L_M(s)ds}-\tilde{\mc W}(\tau)\circ\mathrm{e}^{\sum_{j=1}^M\mc L_j\tau/M}\|_\diamond\notag\\
        &+\left\|\tilde{\mc W}(\tau)\circ\mathrm{e}^{\sum_{j=1}^M\mc L_j\tau/M}-\tilde{\mc W}(\tau)\circ\tilde{\mc M}(\tau)\left(\prod_{l=1}^m\tilde{\mc N}_{l}(\tau)\tilde{\mc J}_{l}(\tau)\right)\tilde{\mc F}(\tau)\right\|_\diamond\notag\\
        \leq&\order{\frac{\tau^2\|\dot{\mc L}\|_\diamond}{M}}+\left(\frac{2\mathrm{e}\|\mc L\|_\p\tau}{K+1}\right)^{K+1}+\order{\|\mc L\|_\p^K\tau^K}=\order{\frac{\tau^2\|\dot{\mc L}\|_\diamond}{M}+\|\mc L\|_\p^K\tau^K},
    \end{align}
    where errors are from Eq.~\eqref{eq:discrete_error_app}, Proposition~\ref{prop:time_dep}, and Eq.~\eqref{eq:step_error}, respectively.

    By concatenating all $T/\tau$ steps, we can find the overall normalization as
    \begin{align}
        \mu=\prod_{t=j}^{T/\tau}\mu_j\leq\exp(\order{\|\mc L\|_\p^2\tau T})=\order{1}.
    \end{align}
    The overall bias, as can be calculated from Lemma~\ref{lm:concate} given that each step runs an LCS formula for a CPTP map, is bounded by
    \begin{align}
        \epsilon_T\leq\prod_{j=1}^{T/\tau}(1+\epsilon_j)-1\leq\order{\frac{\tau T\|\dot{\mc L}\|_\diamond}{M}+\|\mc L\|_\p^K\tau^{K-1}T}.
    \end{align}
    Choosing proper $\tau=\Theta\left(\frac{1}{\|\mc L\|_\p^2T}\right)$ and $M=\Theta\left(\frac{\|\dot{\mc L}\|_\diamond}{\varepsilon\|\mc L\|_\p^2}\right)$, we can make $\epsilon_T\leq\varepsilon/2$.

    The bias analyzed above can be regarded as the distance between the expected output, $\Bar{\langle O\rangle}$, from the LCS formula and the desired value. 
    Nevertheless, we still need to account for the statistical fluctuation.
    The single-shot measured value $\Tr(\rho O)$ is in the range of $[-\mu\|O\|,\mu\|O\|]$ and is sampled for $N$ rounds independently.
    Therefore, we can bound the sampling fluctuation in our average estimation $\tilde{\langle O\rangle}$ based on the Hoeffding bound as
    \begin{align}
        \Pr(\abs{\tilde{\langle O\rangle}-\Bar{\langle O\rangle}}\geq\epsilon_s\|O\|)\leq2\exp(-\frac{N\epsilon^2_s}{2\mu^2}).
    \end{align}
    By setting $N=8\mu^2\log(2/\delta)/\varepsilon^2=\order{\log(1/\delta)/\varepsilon^2}$ and $\epsilon_s=\varepsilon/2$, this guarantees the overall error to be
    \begin{gather}
        \abs{\tilde{\langle O\rangle}-\Tr(\mc T\mathrm{e}^{\int_{0}^T\mc L(t)dt}[\rho_0]O)}\leq\abs{\tilde{\langle O\rangle}-\Bar{\langle O\rangle}}+\abs{\Bar{\langle O\rangle}-\Tr(\mc T\mathrm{e}^{\int_{0}^T\mc L(t)dt}[\rho_0]O)}\leq(\epsilon_s+\epsilon_T)\|O\|\leq \varepsilon\|O\|,
    \end{gather}
    with probability $1-\delta$.

    As for the gate count, we must notice that the whole algorithm implements the short-time time-independent simulation as well as an additional LCS formula for $T/\tau$ steps.
    As analyzed in Theorem~\ref{thm:time-indep}, the first part costs $\order{m(q+K)}$ elementary gates and uses at most two ancilla-qubits.
    The second part can be simply realized according to Fig.~\ref{fig:LCS} with $\order{K}$ gates.
    Given the overall $T/\tau$ steps, here comes the overall gate count as
    \begin{gather}
        \order{m\|\mc L\|_\p^2T^2\left(q+\frac{\ln(1/\varepsilon)}{\ln(\|\mc L\|_\p T)}\right)}.
    \end{gather}
\end{proof}

\section{Qubit Reuse in LCS Realization}\label{sec:reuse}
The implementation of unitary-conjugate operations requires an ancilla qubit with several controlled gates, as illustrated in Fig.~\ref{fig:concept}(c).
When concatenating two arbitrary unitary-conjugate-based formulas, the straightforward approach demands two ancilla qubits (Fig.~\ref{fig:reuse}(a)).
While one might consider measuring and re-preparing the ancilla qubit between operations, such frequent measurements and state preparations are impractical for many experimental platforms.

\begin{figure}[t]
    \centering
    \includegraphics[width=\columnwidth]{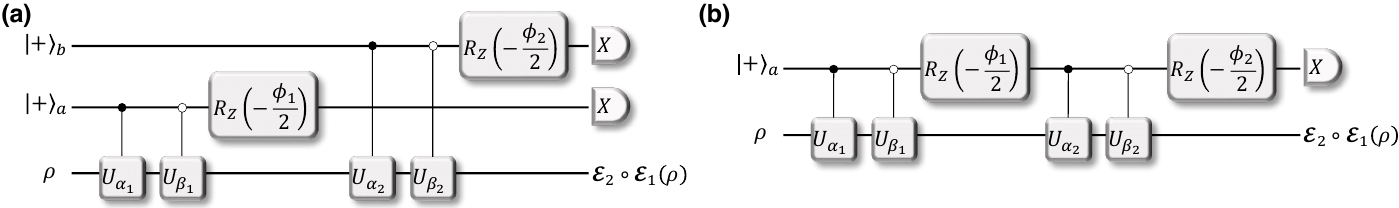}
    \caption{Two equivalent concatenations of the unitary-conjugate LCS formulas.
    (a) The trivial implementation.
    For multiple LCS formulas, this concatenation requires multiple ancilla qubits.
    (b) The concise implementation. 
    Here we check and derive this more concise concatenation using only one ancilla regardless of the number of LCS formulas.}
    \label{fig:reuse}
\end{figure}

We present an efficient approach for concatenating Hermitian-preserving formulas that enables ancilla qubit reuse, as shown in Fig.~\ref{fig:reuse}(b).
To validate this technique, consider the case of two LCS formulas:
\begin{gather}
\mc E_1\coloneqq\mu_1\sum_{\alpha_1,\beta_1}\Pr(\alpha_1,\beta_1|1)\cdot\mc V_{\phi_1,\alpha_1,\beta_1},\ 
\mc E_2\coloneqq\mu_2\sum_{\alpha_2,\beta_2}\Pr(\alpha_2,\beta_2|2)\cdot\mc V_{\phi_2,\alpha_2,\beta_2}.
\end{gather}
The equivalence between both concatenation methods can be verified using two key properties: $\phi(\alpha,\beta)=-\phi(\beta,\alpha)$ and $\Pr(\alpha,\beta)=\Pr(\beta,\alpha)$. 
By reordering the summation indices, we obtain:
\begin{gather*}
\sum_{\substack{\alpha_1,\beta_1\\\alpha_2,\beta_2}}\Pr(\alpha_2,\beta_2|2)\Pr(\alpha_1,\beta_1|1)\mc V_{\phi_2,\alpha_2,\beta_2}\circ\mc V_{\phi_1,\alpha_1,\beta_1}=\sum_{\substack{\alpha_1,\beta_1\\\alpha_2,\beta_2}}\Pr(\alpha_2,\beta_2|2)\Pr(\alpha_1,\beta_1|1)\mathrm{e}^{\ii(\phi_1(\alpha_1,\beta_1)+\phi_2(\alpha_2,\beta_2))}U_{\alpha_2}U_{\alpha_1}\rho U_{\beta_1}^\dag U_{\beta_2}^\dag\notag\\
=\frac{1}{2}\sum_{\substack{\alpha_1,\beta_1\\\alpha_2,\beta_2}}\Pr(\alpha_2,\beta_2|2)\Pr(\alpha_1,\beta_1|1)(\mathrm{e}^{\ii(\phi_1(\alpha_1,\beta_1)+\phi_2(\alpha_2,\beta_2))}U_{\alpha_2}U_{\alpha_1}\rho U_{\beta_1}^\dag U_{\beta_2}^\dag+\mathrm{e}^{-\ii(\phi_1(\alpha_1,\beta_1)+\phi_2(\alpha_2,\beta_2))}U_{\beta_2}U_{\beta_1}\rho U_{\alpha_1}^\dag U_{\alpha_2}^\dag).
\end{gather*}
This mathematical equivalence confirms that the qubit reuse technique shown in Fig.~\ref{fig:reuse}(b) produces the same effective output in expectation.
By induction, this reuse strategy generalizes to longer concatenations of operations, providing a resource-efficient implementation method of unitary-conjugate LCS formulas.

\section{Numerical Details}
Our numerical comparisons in the main text present specific elementary gate counts for different simulation methods.
Here, we elaborate on the compilation and counting procedures for both simulation schemes.
Our compilation focuses primarily on single-qubit rotation gates with arbitrary angles and CNOT gates, as these are the most resource-intensive elementary operations.
For comparison, we have adopted a 20-qubit transversed-field Ising model:
\begin{gather}
    H= J\sum_{\substack{i,j=1\\\langle i,j\rangle}}^{20} {Z_iZ_j}+h\sum_{j=1}^{20}{X_j}.
\end{gather}
The open-system dynamics are introduced through a dissipative term on the first qubit, given by $D=\sqrt{\gamma}\ket{0}\bra{1}$.

\begin{figure}
    \centering
    \includegraphics[width=\linewidth]{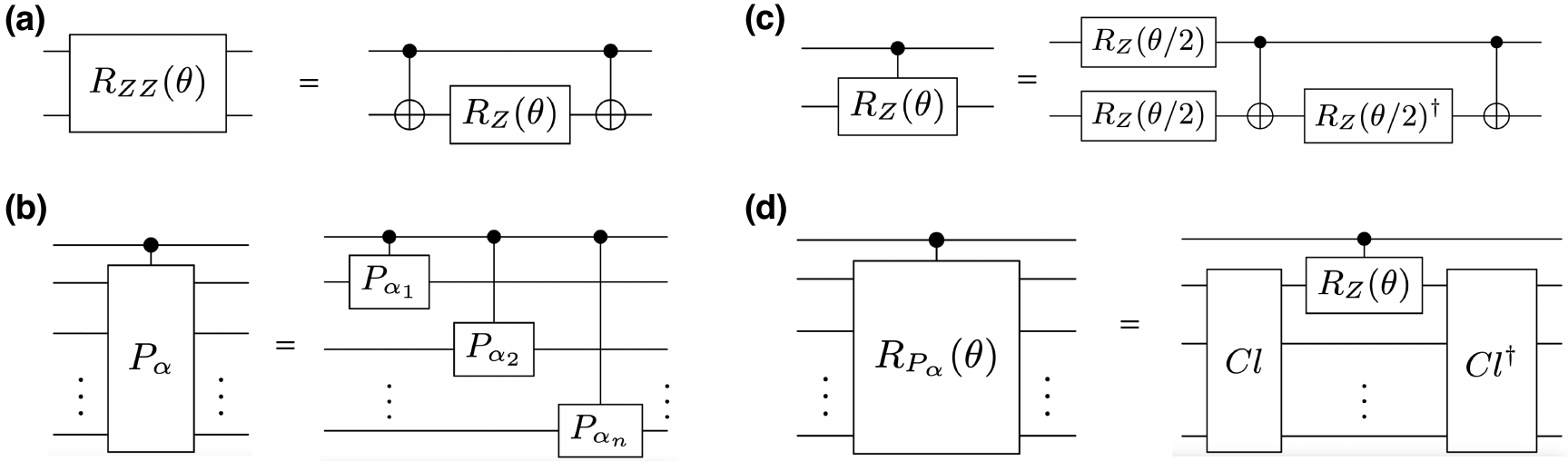}
    \caption{The compilation of typical gates into elementary gates. 
    (a) The compilation of $ZZ$ rotations.
    This can also be applied to arbitrary two-qubit rotation gates by adding single-qubit gates accordingly.
    (b) The compilation of control-Pauli gates.
    The weight of the Pauli operator determines the resulting gate count.
    (c) \& (d) The compilation of the control-rotation gates. 
    The Clifford gates in (d) cost CNOT gates of which the number equals to the weight of $P_\alpha$ as in~\cite{mansky2023decomposition}.}
    \label{fig:gate}
\end{figure}

In the Trotter-based simulation, the scheme is essentially implementing our coarse-grained stage.
For each time step, the method sequentially applies the first-order Trotter-Suzuki formula to both the inherent Hamiltonian $H$ and the joint Hamiltonian $J$ as a proxy of the dissipative evolution:
\begin{gather}
    J\coloneqq\left(\begin{array}{cccc}
0 & D^{\dagger} \\
D & 0 
\end{array}\right)=c(XX+YY).
\end{gather}
The first formula requires 20 single- and two-qubit rotation gates, while the second formula implements two two-qubit rotation gates.
As shown in Fig.~\ref{fig:gate}(a), this method requires 42 $R_Z$ gates and 44 CNOT gates per step.

In our two-stage method, the procedure not only encompasses the same Trotter step as the coarse-grained stage but also embeds LCS compensations.
Recall Fig.~\ref{fig:LCS}, we need two control-unitary gates along with one rotation gate for each LCS formula.
We then follow Eq.~\eqref{eq:Lind_decomp_app} to enumerate resource requirements.
The compensation for the Hamiltonian simulation requires control-rotation gates, and all other compensations ask for control-Pauli gates, which can all be compiled as in Fig.~\ref{fig:gate}.
Note that the weights of the Pauli operator involved are increasing with the order of the compensation $K$.
By calculation, we find that our method requires $48+16K$ CNOT gates and $51$ $R_Z$ gates.

\end{document}